\documentclass[journal,twoside,web]{ieeecolor}
 \usepackage{generic}
\usepackage{cite}
\usepackage{amsmath,amssymb,amsfonts}
\usepackage{algorithm}

 \usepackage{algorithmic}
 \usepackage{graphicx}
 \usepackage{textcomp}
\usepackage{bbm}

 \let\labelindent\relax
\usepackage{enumitem}

\makeatletter
\def\fnum@figure{\textcolor{subsectioncolor}{\sf Fig.~\thefigure}}
\def\fnum@table{\textcolor{subsectioncolor}{\sf TABLE~\thetable}}
\makeatother

 \newtheorem{prob}{Problem}
 \newtheorem{ass}{Assumption}
 \newtheorem{prop}{Proposition}
 \newtheorem{lem}{Lemma}
 \newtheorem{thm}{Theorem}
 \newtheorem{coro}{Corollary}

 \newtheorem{rem}{Remark}
 \newtheorem{cas}{Case}
 \newtheorem{exam}{Example}
  
\usepackage{mathtools}

\usepackage{color}
\makeatletter
\def\zz{\edef\zzz{\pdfliteral{\current@color}}}
\makeatother

\begin{document}
\title{Suboptimality analysis of receding horizon quadratic control with unknown linear systems and its applications in learning-based control }
\author{Shengling Shi, Anastasios Tsiamis, and Bart De Schutter, \IEEEmembership{Fellow, IEEE} 
\thanks{This research has received funding from the
European Research Council (ERC) under the European Union’s Horizon
2020 research and innovation programme (Grant agreement No.101018826 - CLariNet).}
\thanks{ Shengling Shi is with the Department of Chemical Engineering, Massachusetts Institute of Technology, USA, and the Delft Center for Systems and Control, Delft University of Technology, the Netherlands (e-mail: slshi@mit.edu). Anastasios Tsiamis is with the Automatic Control Laboratory, ETH Zurich, Switzerland (e-mail: atsiamis@control.ee.ethz.ch). Bart De Schutter is with the Delft Center for Systems and Control, Delft University of Technology (e-mail: b.deschutter@tudelft.nl).  }  }

\maketitle

\begin{abstract}
 This work analyzes how the trade-off between the modeling error, the terminal value function error, and the prediction horizon affects the performance of a nominal receding-horizon linear quadratic (LQ) controller. By developing a novel perturbation result of the Riccati difference equation, a novel performance upper bound is obtained and suggests that for many cases, the prediction horizon can be either $1$ or $+\infty$ to improve the control performance, depending on the relative difference between the modeling error and the terminal value function error. The result also shows that when an infinite horizon is desired, a finite prediction horizon that is larger than the controllability index can be sufficient for achieving a near-optimal performance, revealing a close relation between the prediction horizon and controllability. The obtained suboptimality performance upper bound is applied to provide novel sample complexity and regret guarantees for nominal receding-horizon LQ controllers in a learning-based setting. We show that an adaptive prediction horizon that increases as a logarithmic function of time is beneficial for regret minimization.
\end{abstract}

\begin{IEEEkeywords}
Receding horizon control, model predictive control, learning-based control.
\end{IEEEkeywords}

\section{Introduction}
\label{sec:introduction}

Receding-horizon control (RHC), also called model predictive control (MPC) interchangeably, has been applied to various applications, such as process control, building climate control, and robotics control \cite{schwenzer2021review, katayama2023model}. It has important advantages such as using optimization based on a model of the system to improve the control performance and the ability to handle constraints and multivariate systems. Due to the importance of RHC in applications, extensive efforts have been devoted to its performance analysis in \cite{mayne2014model, grune2011Book,grune2020economic,kohler2021stability} and references therein, such as stability, constraint satisfaction, and suboptimality analysis. These analyses can help understand the behavior of RHC and motivate novel RHC approaches.

While the majority of works focus on stability and constraint satisfaction \cite{mayne2014model}, this work considers the suboptimality analysis of the closed-loop control performance. We briefly illustrate our problem via a generic RHC scheme. Given a discrete-time system $x_{t+1} = f_\star(x_t,u_t)$, where  $t \in \mathbb{Z}^+$, $x_t \in \mathbb{R}^n$ and $u_{t } \in \mathbb{R}^m$ are the state and input, respectively, we consider the setting where $f_\star$ is unknown, but we have access to an approximate prediction model $\widehat{f}$ instead. At time step $t$, a generic RHC controller has the following form:
\begin{align}
\min_{ \{u_{k|t}\}_{k=0}^{N-1} } & \sum_{k=0}^{N-1} l(x_{k|t},u_{k|t} ) + \widehat{V}(x_{N|t}), \label{eq:GeneralRHC} \\
\text{s.t. }  x_{0|t}& = x_t, \text{ }x_{k+1|t} = \widehat{f}(x_{k|t},u_{k|t} ), \text{ } k\in \{0,\dots,N-1\}, \nonumber
\end{align}
where $N\in \{1,2,\dots\}$ is the prediction horizon, $l$ is a stage cost function, and $x_{k \mid t}$ denotes the predicted state vector at time step $k+t$, obtained by applying the input sequence $\{u_{k\mid t}\}_{k=0}^{N-1}$ to $ \widehat{f}$ starting from $x_t$. In \eqref{eq:GeneralRHC}, a terminal value function $\widehat{V}$ is incorporated. At time step $t$, only the computed $u_{0|t}$ from \eqref{eq:GeneralRHC} is applied as the input $u_t$ to the system.

In \eqref{eq:GeneralRHC}, a modeling error may be present, i.e., $\widehat{f}$ deviates from the unknown real system $f_\star$, which may result from system identification, linearization, or model reduction. While the modeling error can be explicitly considered in the robust controller design\cite{mayne2014model}, the nominal controller \eqref{eq:GeneralRHC} is more fundamental and commonly used in practice. The terminal value function $\widehat{V}$ avoids the short-sightedness of the finite-horizon prediction and thus improves the control performance and guarantees stability. It is standard in RHC \cite{grune2011Book,kohler2021stability} and also advocated in reinforcement learning (RL) to merge RL and RHC for performance improvement \cite{bertsekas2022newton}, where $\widehat{V}$ is trained offline, e.g., by the value iteration \cite{moreno2023predictive}. While the optimal value function $V_\star$ is ideal for achieving stability and the optimal infinite-horizon control performance simultaneously, it is typically intractable to compute it exactly \cite{bertsekas2019reinforcement}. Considering an approximation $\widehat{V}$ of $V_\star$ is then a practical choice.

The suboptimality analysis of \eqref{eq:GeneralRHC} is an open and challenging problem, due to the interplay between $N$ and the error sources. Particularly, when the prediction model is exact, it is well-known that increasing $N$ improves the closed-loop control performance \cite{grune2011Book,bertsekas2019reinforcement}. However, if a modeling error is present, increasing $N$ also propagates the prediction error. On the other hand, having a $\widehat{V}$ closer to the optimal value function $V_\star$ can also improve the control performance. Therefore, the control performance is affected by the complex tradeoff between the error $\|\widehat{V}- V_\star\|$ of the terminal value function, the modeling error $\|\widehat{f} - f_\star\|$, and the prediction horizon $N$, where the norm will be defined later for the setting of this work. The challenges in performance analysis are understanding how the errors propagate over the prediction horizon, how the length of the prediction horizon affects this propagation, and how the error propagation affects the closed-loop control performance.

This work provides a novel suboptimality analysis of \eqref{eq:GeneralRHC} under the joint effect of the modeling error, the prediction horizon, and the terminal value function. We consider the probably most fundamental setting, where the system is linear and the stage cost is a quadratic function, i.e., the nominal receding-horizon LQ controller \cite{bitmead1991riccati}. As shown in our work, the analysis in this setting is already highly challenging. 

Moreover, we apply our analysis to the performance analysis of learning-based RHC controllers. In this case, the model is estimated from data, leading to an inevitable modeling error. Our suboptimality analysis can incorporate this modeling error to provide performance guarantees.

\textit{Related work:} When the model is exact, the suboptimality analysis of RHC controllers, with constraints or economic cost, has been studied extensively in \cite{grune2011Book,grune2020economic,kohler2021stability} and references therein. However, performance analysis in a setting where the system model is uncertain or unknown is rare. The suboptimality analysis of RHC for linear systems with a structured parametric uncertainty is considered in \cite{moreno2022performance}; however, the impact of the approximation in the terminal value function is not investigated. Other relevant works can be found in the performance analysis of learning-based RHC \cite{muthirayan2021online,lale2021model,dogan2021regret}, where the controller actively explores the state space of the unknown system and the model is recursively updated. There, a control performance metric called \textit{regret} is concerned, which measures the accumulative performance difference over a finite time window between the controller and the ideal optimal controller. The impact of the modeling error has been investigated in the above analysis; however, the effect of the prediction horizon and the terminal value function on the control performance is not considered \cite{dogan2021regret,lale2021model,muthirayan2021online}.

As we consider the LQ setting, the performance analysis of LQ regulator (LQR) for unknown systems is relevant, which has received renewed attention from the perspective of learning-based control \cite{recht2019tour,matni2019self,tsiamis2023statistical}. The suboptimality analysis of LQR with a modeling error has been considered in  \cite{dean2020sample,mania2019certainty,simchowitz2020naive}. This analysis is essential for deriving performance guarantees for learning-based LQR. It typically relies on the perturbation analysis of the Riccati equation \cite{mania2019certainty,simchowitz2020naive}, which characterizes the solution of the Riccati equation under a modeling error. However, since LQR concerns an infinite prediction horizon, the above analysis does not consider the effect of the prediction horizon and the terminal value function.
 
\textit{Contribution:} 

In this work, the main results are developed for controllable systems and then generalized to stabilizable systems. The contributions of this work are summarized here:
\begin{enumerate}[wide]
\item As an essential step for achieving performance analysis with unknown systems, the performance analysis in the simpler setting, where the model is exact, is considered first. A novel performance upper bound for the receding-horizon LQ controller is obtained, which reveals a novel relation between the control performance and controllability. This is achieved by a novel convergence analysis of the Riccati difference equation, generalizing the existing results in \cite{del2021note}. Detailed discussions can be found in Section~\ref{sec:RicaKnown}.

  \item To achieve the suboptimality analysis with a modeling error, a core technical result is first established, i.e., a novel convergence rate of the Riccati difference equation when a modeling error is present. This result reveals how the errors propagate over the prediction horizon within the Riccati iterations. It generalizes the perturbation results of the discrete Riccati equation \cite{mania2019certainty,simchowitz2020naive}. While the analyses in \cite{mania2019certainty,simchowitz2020naive} address the modeling error only, the incorporation of the prediction horizon and the terminal value function error in this work leads to major technical challenges.

    \item Based on the above perturbation analysis, a novel performance upper bound for the nominal receding-horizon LQ controller with an inexact model is derived. The derived performance \textit{upper bound} shows how the control performance may vary with changes in the modeling error, the terminal value function error, and the prediction horizon. Particularly, it suggests that for many cases, the prediction horizon can be either $1$ or $+\infty$ to achieve a better performance, depending on the relative difference between the modeling error and the terminal value function error. Moreover, when $N \to +\infty$ is desired, the result suggests that for controllable systems, a finite $N$ that is larger than the controllability index is sufficient for achieving a near-optimal performance.
    
\item The performance upper bound leads to a novel suboptimality bound for the nominal receding-horizon LQ controller, where the unknown system is estimated offline from data. The bound reveals how the control performance depends on the number of data samples. Moreover, a novel regret bound is derived for an adaptive RHC controller. This controller extends the state-of-the-art adaptive LQR controller with $\varepsilon$-greedy exploration \cite{simchowitz2020naive,tsiamis2023statistical} from the infinite prediction horizon to an arbitrary finite prediction horizon. We show a novel regret upper bound $\tilde{O}\big( T \mu_\star^{N}+\sqrt{T} \big)$ for a fixed $N$, where $\mu_\star \in (0,1)$ is a constant and $T$ denotes the total number of time steps for the closed-loop operation, and a regret bound $\tilde{O}\big( \sqrt{T} \big)$ for an adaptive $N$ being a logarithmic function of time.
\end{enumerate}
\vspace*{-0.1cm}

\textit{Outline:} This paper is structured as follows. Preliminaries are introduced in Section~\ref{sec:Problem}. For controllable systems, the suboptimality analysis is developed in Section~\ref{sec:Knownmodel} for an exact model and is generalized in Section~\ref{sec:UnknownModel} to incorporate the modeling error. The above results are extended in Section~\ref{sec:sta} to stabilizable systems and applied to learning-based control in Section~\ref{sec:Application}. All proofs are presented in the Appendix.

\textit{Notation:} Let $\mathbb{S}^n \subseteq \mathbb{R}^{n \times n}$ denote the set of symmetric matrices of dimension $n$, and let $\mathbb{S}^n_+$, $\mathbb{S}^n_{++}$ denote the subsets of positive semi-definite and positive-definite symmetric matrices, respectively. Given $F_1$, $F_2 \in \mathbb{S}^n$, $F_1 \preceq F_2$ denotes $F_2 - F_1 \in \mathbb{S}^n_+$. Moreover, $\mathbb{Z}^+$ denotes the set of non-negative integers. Let $\mathrm{id}$ denote the identity function, and $f \circ g$ denotes the composition of two functions. The symbol $\| \cdot \|$ denotes the spectral norm of a matrix and the $l_2$ norm of a vector. Given any real sequence $\{a_k\}_{k=0}^\infty$, we define $\sum_{k=i}^j a_k \triangleq 0$ for the special case $j<i$. Given a real matrix $F$ and a real square matrix $A$, $\overline{\sigma}(F)$ and $ \underline{\sigma}(F)$ denote the maximum and the minimum singular values, respectively, and $\rho(A)$ denotes the spectral radius. Note that $\overline{\sigma}(F) = \|F\|$. The symbol $\mathcal{U}_{[a,b]}$ denotes the uniform distribution over the interval $[a,b]$, $\mathbb{E}$ denotes the expected value, and $\lfloor\cdot \rfloor$ denotes the floor function. Given $\varepsilon$, $t \geqslant 0$ and two non-negative functions $f(\cdot,\cdot)$ and $g(\cdot,\cdot)$, we write $g(\varepsilon,t) = O(f(\varepsilon,t))$ (as $\varepsilon \to 0$, $t \to \infty$) if $\exists$ $c_1, c_2 >0$ $:$ 
$g(\varepsilon,t) \leqslant c_1 f(\varepsilon,t)$ for any $1/\varepsilon \geqslant c_2$, $t\geqslant c_2$. If $g$ and $f$ depend on $t$ only, $g(t) = \tilde{O}(f(t))$ (as $t \to \infty$) indicates $g(t) = O(f(t) \log^k(t))$ for some $k \in \mathbb{Z}^+$ \cite{cormen2022introduction}.

\section{Preliminaries and problem formulation} \label{sec:Problem}

\subsection{Preliminaries}
We consider the discrete-time linear system
\begin{align}
    x_{t+1} = A_\star x_t + B_\star u_t + w_t, \label{eq:TrueModel}
\end{align}
where $w_t$ is a white noise signal with $\mathbb{E}(w_t w_t^\top) = \sigma_w^2 I$, $\sigma_w>0$, $A_\star \in \mathbb{R}^{n\times n}$, and $B_\star \in \mathbb{R}^{n \times m}$. The initial state $x_0$ is a zero-mean random vector with $\mathbb{E}(x_0 x_0^\top) = X \in \mathbb{S}^n_+$, and it is uncorrelated with the noise. 

In this work, the true system $(A_\star,B_\star)$ is unknown, and we have access to an approximate model $(\hat{A},\hat{B})$ that satisfies $\|\hat{A}- A_\star\| \leqslant \varepsilon_\mathrm{m}$ and $\|\hat{B}- B_\star\| \leqslant  \varepsilon_\mathrm{m}$ for some $\varepsilon_\mathrm{m} \geqslant 0$. This approximate model and its error bound can be obtained, e.g., from linear system identification \cite{simchowitz2018learning,sarkar2019near}. 

In the above setting, we consider the nominal receding-horizon LQ controller \cite{bitmead1991riccati}. At time step $t$, $x_t$ is measured, and the RHC controller solves the following optimization problem:
\begin{subequations} \label{eq:MPCproblem}
\begin{align}
\hspace*{-0.1cm}   \min_{\{u_{k\mid t}\}_{k=0}^{N-1}} \mathbb{E}_{\{w_{k+t}\}_{k=0}^{N-1} } \Big[ \sum_{k=0}^{N-1}l(x_{k\mid t},u_{k\mid t}) + x_{N\mid t}^\top P x_{N\mid t} \Big],\\
    \text{ s.t. } x_{k+1\mid t} = \hat{A} x_{k\mid t} + \hat{B} u_{k\mid t} +w_{k+t}, \text{ }x_{0|t} = x_t,
\end{align}
\end{subequations}
where $l(x_{k\mid t},u_{k\mid t})=x_{k\mid t}^\top Q x_{k\mid t}+ u_{k\mid t}^\top R u_{k \mid t} $, $R\in \mathbb{S}^m_{++}$, and $Q$, $P \in \mathbb{S}^n_+$. Then $u_t = u_{0 \mid t}$ is the input at time step $t$.

We will first consider controllable systems. The extension to stabilizable systems is in Section~\ref{sec:sta}. To this end, we define $ \mathcal{C}_{i}(A_\star,B_\star) \triangleq \begin{bmatrix}
      B_\star & A_\star B_\star &\dots &  A_\star^{i-1} B_\star  
     \end{bmatrix}$ for $i \in \{1,2,\dots\}$ and introduce the following assumptions:
 \begin{ass} \label{ass:ControlIndex}
     There exists $i \in  \{1,2,\dots,n\}$ such that $\mathcal{C}_{i}(A_\star,B_\star)$ has full rank, and the minimum $i$ is denoted by $n_{\mathrm{cr}}$, called the \textit{controllability index}.
 \end{ass}

\begin{ass} \label{ass:LQR}
$A_\star \not=0$, $P \succeq 0 $, $R \succeq I $, and $Q \succeq I$ hold\footnote{Requiring $Q$ to be positive-definite leads to a loss of generality. This assumption, however, is commonly imposed to obtain quantitative bounds\cite{simchowitz2020naive}. Given $Q \succ 0$, $R \succeq I $ and $Q \succeq I$ can be achieved without losing generality by scaling them simultaneously. Letting $A_\star \not=0$ is not restrictive and ensures that $1- \beta_\star$ is invertible for technical simplicity.}.
\end{ass}

For some intermediate technical results, we will relax the controllability requirement whenever possible for generality.
\begin{ass} \label{ass:LQRgeneral2}
System $(A_\star,B_\star)$ is stabilizable.
\end{ass}
 
This RHC controller is further characterized by the Riccati equation. Given any $(A,B) \in \mathbb{R}^{n\times n } \times \mathbb{R}^{n \times m}$, we define \begin{equation} \label{eq:Sb}
S_{B}  \triangleq BR^{-1}B^\top,\end{equation}
and the Riccati mapping $\mathcal{R}_{A,B}: \mathbb{S}_+^n\to \mathbb{S}_+^n$ as
\begin{align}
\mathcal{R}_{A,B}(P) &\triangleq A^\top P (I+S_{B} P)^{-1}A + Q \label{eq:RiccatiMap}\\
&= A^\top P A - A^\top P  B (R+B^\top P B)^{-1} B^\top P   A +Q, \nonumber
\end{align}
where the last equality holds by the Woodbury matrix identity. To apply the Riccati mapping recursively, we define the Riccati iteration $\mathcal{R}^{(i)}_{A,B}$ via recursion: For any integer $i \geqslant 0$,
$$
\mathcal{R}^{(i+1)}_{A,B} \triangleq \mathcal{R}_{A,B} \circ \mathcal{R}^{(i)}_{A,B}, \text{ and }\mathcal{R}^{(0)}_{A,B} = \mathrm{id}.
$$

Then the RHC controller in \eqref{eq:MPCproblem} is equivalent to a linear static controller $u_t=-K_{\mathrm{RHC}}x_t$ \cite{bitmead1991riccati}, where
\begin{equation} 
K_{\mathrm{RHC}}\! = \!\Big[ R+\hat{B}^\top\! \big[\mathcal{R}^{ (N-1)}_{\hat{A},\hat{B}}(P) \big] \hat{B} \Big]^{-1} \!\hat{B}^\top \! \big[\mathcal{R}^{( N-1)}_{\hat{A},\hat{B}}(P) \big]\hat{A}. \label{eq:Kmpc}
\end{equation}

We consider the expected average infinite-horizon performance of $K_{\mathrm{RHC}}$, applied to the true system \eqref{eq:TrueModel}. Given any controller gain $K$ with $\rho(A_\star-B_\star K)<1$, we denote its expected average infinite-horizon performance as 
\begin{align*}
J_K \triangleq \lim_{T\to \infty} & \frac{1}{T} \mathbb{E}_{x_0,\{w_t\}} \Big[ \sum_{t=0}^{T-1}\big( x_t^\top Q x_t+ u_t^\top R u_t \big) \Big], \\
&\text{s.t. } \eqref{eq:TrueModel}, \text{ }u_t = -K x_t, \text{ } t \in \mathbb{Z}^+.
\end{align*}
 
Under a stabilizing $K$, the closed-loop system satisfies \begin{equation}
  \lim_{t \to \infty} \mathbb{E} (x_t x_t^\top) = \Sigma_K, \label{eq:SigmaK}
\end{equation}for some $\Sigma_K  \in \mathbb{S}^n_{++} $ that satisfies the Lyapunov equation:
$$
    \Sigma_K = (A_\star-B_\star K) \Sigma_K (A_\star -B_\star K)^\top + \sigma_w^2 I. 
$$
 
Under Assumption~\ref{ass:LQR}, the optimal controller gain $K_\star$, which achieves the minimum $J_K$, is the LQR controller gain: \begin{equation}\label{eq:Kstar}
 K_\star = (R+B_\star^\top P_\star B_\star)^{-1} B_\star^\top P_\star A_\star,
\end{equation}
where $P_\star$ is the unique positive-definite solution of the discrete-time Riccati equation (DARE) \cite{bertsekas2012dynamic}:
\begin{equation}
P_\star = \mathcal{R}_{A_\star,B_\star}(P_\star). \label{eq:DARE}
\end{equation}
We further define the closed-loop matrix $L_\star$ under $K_\star$ as 
\begin{equation}
L_\star \triangleq A_\star-B_\star K_\star = (I+S_{B_\star} P_\star )^{-1} A_\star . \label{eq:Lstar}
\end{equation}
It is well-known that $L_\star$ is Schur stable, and its decay rate can be characterized by the standard Lyapunov analysis:
\begin{lem} \label{lem:DecayL}
 If Assumptions~\ref{ass:LQR} and \ref{ass:LQRgeneral2} hold, for any $i \in \mathbb{Z}^+$, 
\begin{equation} \label{eq:RateLstar}
\|L_\star^i\| \leqslant   \sqrt{\beta_\star^{-1}}  \Big(\sqrt{ 1-\beta_\star } \Big)^i,
\end{equation} 
with $\beta_\star \triangleq     \underline{\sigma}(Q)/  \overline{\sigma}(P_\star)$, and $ \beta_\star \in (0,1)$ holds.
\end{lem}

The proof of this result is presented in Appendix for completeness. If Assumption~\ref{ass:ControlIndex} holds additionally, then $(A_\star,B_\star)$ is controllable, which implies that $(L_\star,B_\star)$ is controllable, and thus there exists a real number $\nu_\star $ such that \begin{equation} \label{eq:defGramianLowerBound}
\underline{\sigma}\left(\mathcal{C}_{n_{\mathrm{cr}}}(L_\star,B_\star)\right)\geqslant \nu_\star>0.\end{equation}

\subsection{Problem formulation}
In this work, we analyze the performance gap $J_{K_{\mathrm{RHC}}}-J_{K_\star}$ between the RHC controller and the ideal LQR controller. 
\begin{prob} \label{prob}
Given a prediction horizon $N \geqslant 1$, an approximate model $(\hat{A},\hat{B})$ with $\max\{\|\hat{A}- A_\star\|, \|\hat{B}- B_\star\|\} \leqslant  \varepsilon_\mathrm{m}$ and a terminal matrix $P$ with $\|P- P_\star\| \leqslant \varepsilon_\mathrm{p}$ for some $\varepsilon_\mathrm{p} \geqslant 0$, find a non-trivial error bound $g(\varepsilon_\mathrm{m},\varepsilon_\mathrm{p},N)$ such that
$$
J_{K_{\mathrm{RHC}}}-J_{K_\star} \leqslant g(\varepsilon_\mathrm{m},\varepsilon_\mathrm{p},N).
$$
\end{prob}

 In Problem~\ref{prob}, a non-zero $g $ reflects the difference between the performance of the RHC controller, i.e., $J_{K_{\mathrm{RHC}}}$, and the performance of the baseline controller, i.e., $J_{K_\star}$. A non-zero $g$ reveals how the system performance, under the RHC controller, degrades as a function of the errors $\varepsilon_\mathrm{m}$, $\varepsilon_\mathrm{p}$, and the prediction horizon $N$. Here, choosing the infinite-horizon optimal controller as the baseline is standard in studies on learning-based control \cite{mania2019certainty,simchowitz2020naive} and MPC\cite{ grune2011Book}.

In the special case of a known system and $P= P_\star$, we have $K_{\mathrm{RHC}} = K_\star$ and $J_{K_{\mathrm{RHC}}}-J_{K_\star}=0$ for any $N \geqslant 1$. However, in the case of $P \not= P_\star$, which happens when the system is unknown or a numerical error is present for the computed $P_\star$ even if the system is known, we can increase the prediction horizon $N$ to achieve a smaller $J_{K_{\mathrm{RHC}}}-J_{K_\star}$. In the more complex situation where there is a modeling error, a larger $N$ leads to the propagation of the modeling error and thus may enlarge the performance gap. Our target is to characterize the above complex tradeoff among $N$, $\varepsilon_{\mathrm{m}}$, and $\varepsilon_{\mathrm{p}}$.

In this work, the obtained bounds hold regardless of whether $\varepsilon_\mathrm{m}$ and $\varepsilon_\mathrm{p}$ are coupled. The presence of coupling, e.g., $\varepsilon_\mathrm{p} = h(\varepsilon_\mathrm{m})$ for some function $h$, can be easily incorporated by plugging function $h$ into the bound $g$. In addition, the error bounds obtained will also depend on the system matrices $A_\star$ and $B_\star$. To simplify the algebraic expressions of the bounds, we upper bound the system matrices as
$$
\max\{\|A_\star\|, \|B_\star\|,\|P_\star\|  \} \leqslant \Upsilon_\star,
$$
where $\Upsilon_\star \geqslant \max\{1, \varepsilon_\mathrm{m},\varepsilon_\mathrm{p} \}$ is a real constant. Note that there always exists a sufficiently large $\Upsilon_\star $ such that the above holds.

\begin{rem} Regarding the bound $g$ in Problem~\ref{prob}, its rate of change as $N$, $\varepsilon_{\mathrm{m}}$, and $\varepsilon_{\mathrm{p}}$ change will be the primary interest. The other constants involved in the bound will often be simplified for interpretability. Therefore,  the derived bound is more suitable for convergence analysis and qualitative insights instead of being used as a practical error bound.
\end{rem}
\begin{rem}
 We also note that $g$ is a worst-case bound that holds for any $(\hat{A},\hat{B}) \in \mathcal{S} (\varepsilon_\mathrm{m})\triangleq  \{(A,B) \mid \max\{  \|A-A_\star\|, \|B-B_\star\| \} \leqslant \varepsilon_{\mathrm{m}} \}$. Therefore, it can also be conservative for a particular $(\hat{A},\hat{B})$ instance.   
\end{rem}
 \section{Suboptimality analysis with known model} \label{sec:Knownmodel}
 As the first step, we consider the simpler setting where $(\hat{A},\hat{B}) = (A_\star,B_\star)$. The goal is to analyze the joint effect of the prediction horizon and the approximation in the terminal value function on the control performance. To obtain the final performance upper bound in Section~\ref{sec:PerCo}, we first develop two technical results in Sections~\ref{meta} and \ref{sec:RicaKnown}.
\begin{cas} \label{ass:KnownModel}
It holds that $(\hat{A},\hat{B}) = (A_\star,B_\star)$.
\end{cas}

\subsection{Preparatory performance analysis} \label{meta}
We first characterize the performance gap between the LQR controller and a general linear controller that has a similar structure to the RHC controller \eqref{eq:Kmpc}. Given any $(A,B) \in \mathbb{R}^{n\times n} \times \mathbb{R}^{n\times m}$, we define function $\mathcal{K}_{A,B}: \mathbb{S}_+^n \to \mathbb{R}^{m \times n}$ as
\begin{equation} \label{eq:Kpara}
    \mathcal{K}_{A,B}(F) \triangleq  (R+B^\top F B )^{-1} B^\top F A.
\end{equation}
Then given the controller gain $\mathcal{K}_{A_\star,B_\star}(F)$ for any $F \in \mathbb{S}_+^n$, the following result characterizes its performance gap based on the difference $\|F-P_\star\|$.
\begin{lem} \label{lem:ManiaDiffK}
Given any $F \in \mathbb{S}_+^n$ with $\|F-P_\star\| \leqslant \varepsilon$ for some $\varepsilon \geqslant 0$, consider the controller gain $K= \mathcal{K}_{A_\star,B_\star}(F)$ and the optimal controller gain $K_\star$. In Case~\ref{ass:KnownModel}, if Assumptions~\ref{ass:LQR} and \ref{ass:LQRgeneral2} hold, and if $ \varepsilon \leqslant \min\{\Upsilon_\star,  \underline{\sigma}(R)/(40 \Upsilon_\star^4 \|P_\star\|^{3/2})  \}$, then the closed-loop matrix $A_\star-B_\star K$ is Schur stable and
$$
J_K- J_{K_\star} \leqslant c_\star \varepsilon^2, \text{ where }
$$ 
\begin{equation}  \label{eq:cstar}
c_\star \triangleq  \frac{ 32 \min\{n,m\}  \sigma_w^2  (\overline{\sigma}(R)+\Upsilon_\star^3 )\Upsilon_\star^4 \|P_\star\| ( \sqrt{\|P_\star\|}+1)^2 }{ \underline{\sigma}^2(R)}.
\end{equation}
\end{lem}

Lemma~\ref{lem:ManiaDiffK} directly follows from Lemma~\ref{lem:metalemma}, which is presented later, as a special case of a known model. These results extend \cite[Thm. 1]{mania2019certainty} and \cite[Prop. 7]{simchowitz2020naive} by incorporating the terminal value function error. Lemma~\ref{lem:ManiaDiffK} shows that if $F$ is sufficiently close to $P_\star$, the controller gain $\mathcal{K}_{A_\star,B_\star}(F)$ can stabilize the system with a near-optimal performance.

In Case~\ref{ass:KnownModel},
since the RHC controller gain \eqref{eq:Kmpc} is a special case of \eqref{eq:Kpara} with $F = \mathcal{R}^{(N-1)}_{A_\star,B_\star}(P)$, we will exploit Lemma~\ref{lem:ManiaDiffK} to characterize the performance gap of $K_\mathrm{RHC}$. Particularly, Lemma~\ref{lem:ManiaDiffK} can be used to bound $J_{K_{\mathrm{RHC}}}-J_{K_\star} $ if we can quantify the difference $$\|\mathcal{R}^{(N-1)}_{A_\star,B_\star}(P) - P_\star  \|.$$Note that this difference captures the convergence of $\mathcal{R}^{(N-1)}_{A_\star,B_\star}(P)$ to $P_\star$ as $N$ increases.

\subsection{Convergence rate of Riccati iterations} \label{sec:RicaKnown}
It is a classical result that $\mathcal{R}^{(i)}_{A_\star,B_\star}(P)$ converges to $P_\star$ exponentially as $i$ increases \cite{anderson1979optimal}. Different convergence analyses have been investigated in the literature, leading to different convergence rates. The rates in \cite{lawson2007birkhoff} and \cite[Thm. 14.4.1]{hassibi1999indefinite} are complex functions of the system parameters, characterized by the Thompson metric in \cite{lawson2007birkhoff} and a norm weighted by a special matrix in \cite[Thm. 14.4.1]{hassibi1999indefinite}. These convergence rates are hard to interpret because of the specific metric and norm. On the other hand, the rates in \cite[Ch. 4.4]{anderson1979optimal} and \cite{del2021note} utilize the standard matrix 2-norm and thus are easy to interpret; however, \cite[Ch. 4.4]{anderson1979optimal} does not exploit controllability, leading to a slower convergence rate than \cite{del2021note}. To the best of our knowledge, \cite{del2021note} provides the most recent state-of-the-art analysis of convergence rates for controllable systems. Therefore, we develop our results based on \cite{del2021note}. While the result in \cite{del2021note} considers only a single convergence rate, we generalize it to incorporate two distinct convergence rates depending on the range of $i$. We also compute the rates as a function of $\beta_\star$ and thus of $P_\star$. The explicit dependency of the convergence rates on $P_\star$ will facilitate the subsequent perturbation analysis in Section~\ref{sec:UnknownModel}, where the effect of the modeling error on $P_\star$ will be exploited.

To this end, we first define several new notations and functions. Recall $S_B$ defined in \eqref{eq:Sb}, and we define the function $\mathcal{L}_{A,B}: \mathbb{S}_+^n \to \mathbb{R
}^{n\times n} $ as
\begin{equation}
\mathcal{L}_{A,B}(P) \triangleq (I+S_{B} P)^{-1} A = A - B \mathcal{K}_{A,B}(P),  \label{eq:ClosedLoopL}
\end{equation}
which denotes the closed-loop matrix under the controller gain $\mathcal{K}_{A,B}(P)$, and thus $\mathcal{L}_{A_\star,B_\star}(P_\star)  = L_\star$. In addition, for any integers $ 0 \leqslant j \leqslant i$, define 
\begin{equation} \label{eq:Phi}
\Phi^{(j:i)}_{A,B}(P) \triangleq \begin{cases}
\mathcal{L}_{A,B} \big( P^{(j)} \big)   \dots \mathcal{L}_{A,B}\big( P^{(i-1)} \big), & i > j \geqslant 0,\\
I,& i=j,
\end{cases}
\end{equation}
where $P^{(i)}$ is a shorthand notation for $ \mathcal{R}^{(i)}_{A,B}(P)$. The definition shows $\Phi^{(0:i)}_{A_\star,B_\star}(P_\star)  = L_\star^i$. Note that $\Phi^{(0:i)}_{A,B}(P)$ is a state-transition matrix of a time-varying linear system, whose asymptotic stability is reflected by the convergence of $\Phi^{(0:i)}_{A,B}(P) $ to $0$ as $i\to \infty$ \cite{rugh1996linear}.
 
The state-transition matrix is helpful for characterizing $\| \mathcal{R}^{(N-1)}_{A_\star,B_\star}(P) - P_\star \|$ for any $P \in \mathbb{S}_+^n$ \cite{del2021note}:  
\begin{equation}
\mathcal{R}^{(i)}_{A_\star,B_\star}(P) - P_\star = \big[ \Phi^{(0:i)}_{A_\star,B_\star}(P) \big]^\top  (P- P_\star ) L_\star^i, \label{eq:RiccatiDiffSameModel}
\end{equation}
which can be verified by applying Lemma~\ref{lem:RicaDiffEqu}\ref{eq:IdentiRicatiDiif} recursively. Recall that $L_\star$ is Schur stable, satisfying $\lim_{i \to \infty} L^i_\star=0$. Therefore, \eqref{eq:RiccatiDiffSameModel} shows that as $i \to \infty$, the convergence of $\mathcal{R}^{(i)}_{A_\star,B_\star}(P)$ to $P_\star$ is governed by the convergence of $L_\star^i$. Furthermore, exploiting the exponential decay of $ \Phi^{(0:i)}_{A_\star,B_\star}(P)$ can further contribute to the convergence of $\mathcal{R}^{(i)}_{A_\star,B_\star}(P)$.

Overall, we can characterize the rate of convergence of $\mathcal{R}^{(i)}_{A_\star,B_\star}(P) $ to $P_\star$ as follows:
\begin{lem} \label{lem:BasicRate}
Given $P$ with $\|P- P_\star\| \leqslant \varepsilon_\mathrm{p}$, if Assumptions~\ref{ass:ControlIndex} and \ref{ass:LQR} are satisfied, then for $i \in \mathbb{Z}^+$,  
\begin{align} \label{eq:KnownModel}
 &\|\mathcal{R}^{(i)}_{A_\star,B_\star}(P) - P_\star \| \nonumber\\
& \leqslant \begin{cases}
    \overline{\tau}_\star  \left[  \sqrt{ \left(1-\beta_\star \right)\left(  1 -  \overline{\beta}(\varepsilon_{\mathrm{p}})\right) }  \right] ^i \varepsilon_\mathrm{p}, &\text{ if } i < n_{\mathrm{cr} } ,\\
   \tau_\star  \left( 1 - \beta_\star \right)^i \varepsilon_\mathrm{p}, & \text{ if } i \geqslant n_{\mathrm{cr} },
   \end{cases}
\end{align}
with constants $\tau_\star$ and $\overline{\tau}_\star$ defined in \eqref{eq:DefTau} and \eqref{eq:DefTauBar}, respectively, $\beta_\star$ defined in Lemma~\ref{lem:DecayL}, and  
\begin{equation} \label{eq:DefSlowRate}
\overline{\beta}(\varepsilon_{\mathrm{p}})  \triangleq \underline{\sigma}(Q)/ \left( 
\overline{\sigma}(P_\star)+ \varepsilon_{\mathrm{p}}\beta_\star^{-1} \right),
\end{equation}
 where $\overline{\beta}(\varepsilon_{\mathrm{p}}) \leqslant \beta_\star$ holds for any $\varepsilon_{\mathrm{p}} \geqslant 0$. 
\end{lem}

Since $\overline{\beta} \leqslant  \beta_\star$, we have $ 1 - \beta_\star \leqslant 
 \sqrt{ (1-\beta_\star)(1 -  \overline{\beta} ) } $. Therefore, as $i$ increases, Lemma~\ref{lem:BasicRate} shows a slower decay rate\footnote{Given $c_1 a^i$ and $c_2b^i$ with $0 \leqslant a\leqslant b<1$ and $c_1$, $c_2 \geqslant 0$, $i \in \mathbb{Z}^+$, we say $a$ is a faster decay rate than $b$ to indicate $a^i$ decays faster than $b^i$ as $i$ increases.} $\sqrt{ (1-\beta_\star)(1 -  \overline{\beta} ) }$ for $i <n_{\mathrm{cr}}$ and a faster rate $1 - \beta_\star$ for $i \geqslant n_{\mathrm{cr}  }$. Recall $\Phi^{(0:i)}_{A_\star,B_\star}(P)$ in \eqref{eq:Phi} and \eqref{eq:RiccatiDiffSameModel}. Intuitively, the case $i  < n_{\mathrm{cr}} $ is when the state-transition matrix $\Phi^{(0:i)}_{A_\star,B_\star}(P)$ is affected by the initial matrix $P$ and the consequent transient behaviors, which is evident from the dependency of $\overline{\beta}$ on $\varepsilon_{\mathrm{p}}$. Then, if $i \geqslant n_{\mathrm{cr}}$, the time-varying closed-loop systems in \eqref{eq:Phi} get closer to $L_\star$, leading to a faster convergence rate. 
 
The technical challenge in developing Lemma~\ref{lem:BasicRate} is the derivation of the two distinct decay rates. Given controllability, the faster rate for $i \geqslant n_{\mathrm{cr}  }$ is established from the full rank property of $\mathcal{C}_{i}(A_\star,B_\star) $ by exploiting Lemma~\ref{lem:DecayL} and \cite[Lem. 4.1]{del2021note}. Since $\mathcal{C}_{i}(A_\star,B_\star) $ is not of full rank for $i < n_{\mathrm{cr}  }$, the slower rate is established from stabilizability. The analysis is achieved by utilizing the stability analysis of linear time-varying systems \cite{rugh1996linear}, introduced in Lemma~\ref{lem:Stability}, and further deriving the decay rate as an explicit function of $\beta_\star$. Note that if $\varepsilon_{\mathrm{p}}=0$, we obtain the rate $(1 -\beta_\star)$ in both cases of \eqref{eq:KnownModel}, recovering the square of the rate $ \sqrt{1 -\beta_\star}$ in \eqref{eq:RateLstar}.

\subsection{Performance upper bound with a known model} \label{sec:PerCo}
The following main result is a direct consequence of \eqref{eq:Kmpc}, Lemmas~\ref{lem:ManiaDiffK} and \ref{lem:BasicRate}:
\begin{thm} \label{thm:KnownModelSub}
In Case~\ref{ass:KnownModel}, consider $K_\mathrm{RHC}$ in \eqref{eq:Kmpc} with any $P \in \mathbb{S}_+^n$ satisfying $\|P-P_\star \| \leqslant \varepsilon_\mathrm{p}$. It holds that $A_\star - B_\star K_\mathrm{RHC}$ is Schur stable with $J_{K_\mathrm{RHC}}-J_{K_\star}  \leqslant v(N, \varepsilon_{\mathrm{p}})  $, where  
\begin{align} \label{eq:JgapKnown}
 & v(N, \varepsilon_{\mathrm{p}})  \nonumber \\ &  \triangleq  \begin{cases} 
c_\star \overline{\tau}^2_\star \left[ (1-\beta_\star)\left(1 -  \overline{\beta}(\varepsilon_{\mathrm{p}}) \right)\right]^{N-1}\varepsilon_\mathrm{p}^2, & \text{ if }   N \leqslant n_{\mathrm{cr}},    \\
c_\star \tau^2_\star  (1-\beta_\star)^{2(N-1)} \varepsilon_\mathrm{p}^2, &\text{  if  } N \geqslant n_{\mathrm{cr}}+1,
\end{cases} 
\end{align}
if Assumptions~\ref{ass:ControlIndex} and \ref{ass:LQR} hold, and if $N$ and $\varepsilon_{\mathrm{p}}$ satisfy $ \sqrt{v(N, \varepsilon_{\mathrm{p}})/c_\star} \leqslant \min\{\Upsilon_\star,  \underline{\sigma}(R)/(40 \Upsilon_\star^4 \|P_\star\|^{3/2})  \} $.
\end{thm}

In Theorem~\ref{thm:KnownModelSub}, the upper bound on $\sqrt{v(N, \varepsilon_{\mathrm{p}})/c_\star}$ means that $N$ should be sufficiently large or $\varepsilon_{\mathrm{p}}$ should be sufficiently small, such that the RHC controller can stabilize the system with a suboptimal performance. Moreover, the performance gap \eqref{eq:JgapKnown} captures the tradeoff between the error $\varepsilon_\mathrm{p}$ and the prediction horizon $N$: The performance gap decays quadratically in $\varepsilon_\mathrm{p}$ and exponentially in $N$. The exponential decay obtained here is analogous to the exponential decay rate in \cite{zhang2021regret} for linear RHC with estimated additive disturbances over the prediction horizon and in \cite{kohler2021stability} for nonlinear RHC, while we have a novel characterization with two distinct decay rates.

The case $N \geqslant n_{\mathrm{cr}}+1$ in \eqref{eq:JgapKnown} has a faster decay rate than the rate for $N \leqslant n_{\mathrm{cr}}$, which indicates the advantage of choosing a horizon larger than the controllability index. This also suggests that we may choose a smaller horizon for controlling fully-actuated systems, with $\mathrm{rank}(B)=n$ and $n_{\mathrm{cr}}=1$, compared to under-actuated systems.

\section{Suboptimality analysis with unknown model} \label{sec:UnknownModel}
In this section, the results in Section~\ref{sec:Knownmodel} are generalized to consider the joint performance effect of the modeling error, terminal matrix error, and prediction horizon. To obtain the final performance upper bound in Section~\ref{sec:PerError}, we first develop several technical results in Sections~\ref{sec:MetaError} and \ref{sec:RiccPertur}.

\subsection{Preparatory performance analysis} \label{sec:MetaError}
Given any $F \in \mathbb{S}_+^n$ and a general controller gain $\mathcal{K}_{\hat{A},\hat{B}}(F) $ in \eqref{eq:Kpara}, formulated on the approximate model, the following result characterizes the performance gap between  $\mathcal{K}_{\hat{A},\hat{B}}(F) $ and the optimal controller gain $K_\star$:
\begin{lem} \label{lem:metalemma}
Given any $F \in \mathbb{S}_+^n$ with $\|F-P_\star \| \leqslant \varepsilon$ for some $\varepsilon \geqslant 0$, consider the controller gain $K = \mathcal{K}_{\hat{A},\hat{B}}(F) $. Then $A_\star -B_\star K$ is Schur stable and  
\begin{equation*}  
J_{K} - J_{K_\star} \leqslant  c_\star (\varepsilon_\mathrm{m} +\varepsilon)^2,
\end{equation*}
with $c_\star$ defined in \eqref{eq:cstar}, if Assumptions~\ref{ass:LQR} and \ref{ass:LQRgeneral2} hold, and if $\Upsilon_\star \geqslant \varepsilon$ and $ 8 \Upsilon_\star^4(\varepsilon_\mathrm{m}+ \varepsilon ) /\underline{\sigma}(R) \leqslant 1/5 \|P_\star\|^{-3/2}$.
\end{lem}

The above result is analogous to \cite[Thm. 1]{mania2019certainty} and \cite[Prop. 7]{simchowitz2020naive}, with extensions to incorporate the additional error term $\|F - P_\star\|$. Lemma~\ref{lem:metalemma} shows that if the error $\varepsilon_\mathrm{m}+\varepsilon$ is sufficiently small, the resulting controller gain $K = \mathcal{K}_{\hat{A},\hat{B}}(F) $ can stabilize the real system $(A_\star,B_\star)$ with a suboptimal performance gap.

\subsection{Perturbation analysis of Riccati difference equation} \label{sec:RiccPertur}
As the RHC controller \eqref{eq:Kmpc} is a special case of $\mathcal{K}_{\hat{A},\hat{B}}(F) $ with $F = \mathcal{R}^{( N-1)}_{\hat{A},\hat{B}}(P) $, Lemma~\ref{lem:metalemma} can be exploited to analyze the suboptimality of the RHC controller. Particularly, Lemma~\ref{lem:metalemma} shows that $J_{K_{\mathrm{RHC}}} - J_{K_\star}$ can be upper bounded by a function of the modeling error bound $\varepsilon_\mathrm{m}$ and $\varepsilon$ satisfying
\begin{equation} \label{eq:RicciDiffPstar}
\big\| \mathcal{R}^{( N-1)}_{\hat{A},\hat{B}}(P) - P_\star  \big\| \leqslant \varepsilon.
\end{equation}
Therefore, the main challenge is to characterize $\varepsilon$ by investigating the Riccati iterations using the approximate model.

The above problem has been studied in the classical perturbation analysis of the Riccati difference equation \cite{konstantinov1995conditioning}; however, the final bound in \cite{konstantinov1995conditioning} does not explicitly show the dependence on $N$, $P_\star$, and $\varepsilon_{\mathrm{p}}$. In this work, we provide a novel upper bound as in \eqref{eq:RicciDiffPstar} that captures this dependence.

To this end, due to the approximate model, we first define a state-transition matrix similar to \eqref{eq:Phi}:
\begin{equation} \label{eq:PhiBar}
\bar{\Phi}^{(j:i)}(P) \triangleq \begin{cases}  \bar{\mathcal{L}}\big(\mathcal{R}^{(j)}_{A_\star,B_\star}(P)\big)  \cdots \bar{\mathcal{L}}\big(\mathcal{R}^{(i-1)}_{A_\star,B_\star}(P)\big), &  i > j \geqslant 0,\\
I, & i=j,
\end{cases}
\end{equation}
where
\begin{equation*}  
\bar{\mathcal{L}}(P) \triangleq \mathcal{W}(P) [ \mathcal{H}(P) + \mathcal{L}_{A_\star,B_\star}(P)],
\end{equation*}
with the functions $\mathcal{W}$ and $\mathcal{H}$ defined on $\mathbb{S}_+^n$:
\begin{equation*}  
\mathcal{W}(P) \triangleq  I+(S_{B_\star}-S_{\hat{B}})P,\text{ }  \mathcal{H}(P) \triangleq (I+ S_{B_\star} P)^{-1}(\hat{A}-A_\star).
\end{equation*}
Recall the closed-loop matrix $\mathcal{L}_{A_\star,B_\star}(P)$ in \eqref{eq:ClosedLoopL}, and thus, we can interpret $\bar{\mathcal{L}}(P)$ as a perturbed $\mathcal{L}_{A_\star,B_\star}(P)$: If there is no modeling error, we have $\bar{\mathcal{L}}(P)= \mathcal{L}_{A_\star,B_\star}(P)$ due to $\mathcal{H}(P)=0$ and $\mathcal{W}(P)=I$. Therefore, $\bar{\Phi}^{(j:i)}(P)$ can also be interpreted as a perturbed version of the state-transition matrix $ \Phi_{A_\star,B_\star}^{(j:i)}(P)$ due to the approximate model. 

Then, as an extension of \eqref{eq:RiccatiDiffSameModel} to the approximate model, the following lemma can be obtained:
\begin{lem} \label{lem:Pdiff}
For any $P_1$, $P_2 \in \mathbb{S}^n_+$ and $i \in \mathbb{Z}^+$, define $\hat{P}^{(i)} \triangleq \mathcal{R}^{(i)}_{\hat{A},\hat{B}}(P_1)$ and $P^{(i)}_\star \triangleq \mathcal{R}^{(i)}_{A_\star,B_\star}(P_2) $. It holds that
\begin{subequations} \label{eq:lemmaRiccati}
\begin{align}
&\hat{P}^{(i)}  - P_\star^{(i)} = \left[ \Phi_{\hat{A},\hat{B}}^{(0:i)}(P_1) \right]^\top \left(P_1 - P_2\right) \bar{\Phi}^{(0:i)}(P_2) \label{eq:PdiffTerm1}\\& +\sum_{j=1}^i  \left[ \Phi_{\hat{A},\hat{B}}^{(i-j+1:i)}(P_1) \right]^\top  \mathcal{M}\left( \hat{P}^{(i-j)}  ,P_\star^{(i-j)} \right) \bar{\Phi}^{(i-j+1:i)}(P_2), \label{eq:PdiffTerm2}
\end{align}
\end{subequations}
where $\mathcal{M}$ is a function on $\mathbb{S}_+^n \times \mathbb{S}_+^n $: $\mathcal{M}(P_1,P_2) \triangleq \hat{A}^\top P_2 (I+S_{B_\star}P_2)^{-1}\hat{A}- A_\star^\top P_2 (I+S_{B_\star}P_2)^{-1}A_\star + \hat{A}^\top (I+P_1 S_{\hat{B}} )^{-1} P_2 (S_{B_\star}-S_{\hat{B}})P_2 (I+S_{B_\star} P_2)^{-1}\hat{A}. $
\end{lem}

Equation \eqref{eq:lemmaRiccati} is closely related to \eqref{eq:RiccatiDiffSameModel}: While \eqref{eq:RiccatiDiffSameModel} concerns the state-transition matrix $\Phi_{A_\star,B_\star}$ of the true system, \eqref{eq:PdiffTerm1} contains the perturbed state-transition matrix $\bar{\Phi}$ and the state-transition matrix $\Phi_{\hat{A},\hat{B}}$ of the approximate model; moreover, \eqref{eq:PdiffTerm2} contains an extra term caused by the modeling error.

To upper bound $\| \mathcal{R}^{( N-1)}_{\hat{A},\hat{B}}(P) - P_\star  \|$, we exploit \eqref{eq:lemmaRiccati} with $P_1 = P$ and $P_2 = P _\star$. This motivates us to analyze the spectral norm of the perturbed state-transition matrices in \eqref{eq:lemmaRiccati}. The analysis of these state-transition matrices is analogous to the stability analysis of the corresponding time-varying systems.
 
Then we establish the following results for the two perturbed state-transition matrices.

\begin{lem} \label{lem:PhiBar}
Given any two integers $i  \geqslant j \geqslant 0$ and $\bar{\Phi}^{(j:i)}$ defined in \eqref{eq:PhiBar}, if Assumptions~\ref{ass:LQR} and \ref{ass:LQRgeneral2} hold, then
\begin{equation} \label{eq:PhiBarDecayRate}
    \|\bar{\Phi}^{(j:i)}( P_\star)\| \leqslant \sqrt{\beta_\star^{-1} } [\gamma_1(\varepsilon_\mathrm{m})] ^{i-j},
\end{equation}
with
\begin{align}
 \label{eq:Defgamma1Psi}
\gamma_1(\varepsilon_\mathrm{m}) & \triangleq \sqrt{\beta_\star^{-1} }  \psi(\varepsilon_\mathrm{m} ) +\sqrt{ 1-\beta_\star },\end{align}
where $\psi$ is defined in \eqref{eq:psiDEFproof} and satisfies $\psi(\varepsilon_\mathrm{m}) = O(\varepsilon_{\mathrm{m}})$.
\end{lem}

Comparing \eqref{eq:PhiBarDecayRate} and \eqref{eq:RateLstar}, $\bar{\Phi}^{i:j}(P_\star)$ admits a perturbed decay rate $\gamma_1$: In $\gamma_1$, $\sqrt{1-\beta_\star}$ from \eqref{eq:RateLstar} is perturbed by the term $\sqrt{\beta_\star^{-1} }  \psi(\varepsilon_\mathrm{m} ) $.

To analyze $\Phi_{\hat{A},\hat{B}}^{(0:i)}$ in \eqref{eq:lemmaRiccati}, we limit the modeling error by introducing the following technical assumptions:
\begin{ass} \label{ass:TechContPertur}
The modeling error $\varepsilon_{\mathrm{m}}$ is sufficiently small such that
\begin{enumerate}[label=(\alph*)] 
         \item $\varepsilon_\mathrm{m} \leqslant 1/(16 \|P_\star\|^2 )$; \label{ass4:i}

        \item $ f_{\mathcal{C}}(\varepsilon_{\mathrm{m}} ) \leqslant \nu_\star/2$, where $\nu_\star$ is defined in \eqref{eq:defGramianLowerBound}, and $f_{\mathcal{C}}$ is defined in \eqref{eq:defFc} and satisfies $f_{\mathcal{C}}(\varepsilon_{\mathrm{m}} ) = O(\varepsilon_\mathrm{m})$; \label{ass4:ii}

             \item $\varepsilon_\mathrm{m} \leqslant \alpha_\star$ with the constant $\alpha_\star$ defined in \eqref{eq:DualPerturbCondi}. \label{ass4:iii}

\end{enumerate}
\end{ass}

\begin{ass} \label{ass:Dual}
 The true system $(A_\star,B_\star,Q,R)$ satisfies  $\underline{\sigma}(P_{\mathrm{dual} }) \geqslant 1$, where $P_{\mathrm{dual}} \in \mathbb{S}^n_{++}$ is a constant matrix formulated on the true system and is defined in \eqref{eq:DualRicca}.  
\end{ass}

Assumption~\ref{ass:TechContPertur} limits the modeling error, and Assumption~\ref{ass:Dual} is related to the controllability of the real system, see more details in Remark~\ref{rem:Dual} of Appendix~\ref{sec:Appendix2}. Assumption~\ref{ass:TechContPertur}\ref{ass4:i} ensures the stabilizability of any $(\hat{A},\hat{B}) \in \mathcal{S} (\varepsilon_\mathrm{m})$ and is used to establish a slower decay rate of $\Phi_{\hat{A},\hat{B}}^{(0:i)}$ for $i < n_{\mathrm{cr}}$, analogous to the slower rate in \eqref{eq:KnownModel}. Moreover, Assumption~\ref{ass:TechContPertur}\ref{ass4:ii} ensures the controllability\footnote{Even if Assumption~\ref{ass:TechContPertur}\ref{ass4:ii} also implies stabilizability, the upper bound for $\varepsilon_{\mathrm{m}}$ in Assumption~\ref{ass:TechContPertur}\ref{ass4:i} is still needed for computing the decay rates.} of $(\hat{A},\hat{B}) $. Then, by combining Assumption~\ref{ass:TechContPertur}\ref{ass4:ii} with Assumptions~\ref{ass:TechContPertur}\ref{ass4:iii} and \ref{ass:Dual}, we can exploit controllability to establish a faster decay rate for $i \geqslant n_{\mathrm{cr}}$, analogous to the faster rate in \eqref{eq:KnownModel}. Note that in Assumptions~\ref{ass:TechContPertur}\ref{ass4:ii}, \ref{ass:TechContPertur}\ref{ass4:iii} and \ref{ass:Dual}, the constants $\nu_\star$, $\alpha_\star$, and $P_{\mathrm{dual}}$ depend on the true system only and are independent of the errors and $N$ under study. See more details in Appendix~\ref{sec:Appendix2}.  

With the above assumptions, we have the following result:
\begin{prop} \label{prop:PhiBarAB}
If Assumptions~\ref{ass:ControlIndex}, \ref{ass:LQR}, \ref{ass:TechContPertur}, and \ref{ass:Dual} hold, then for any $i, j \in \mathbb{Z}^+$, it holds that  
\begin{IEEEeqnarray}{c}\label{eq:PhihatIneq}
\hspace{-0.3cm} \|\Phi_{\hat{A},\hat{B}}^{(j:i)}(P)\| \leqslant 
  \begin{cases}
 \zeta     [\gamma_2(\varepsilon_\mathrm{m})] ^{i-j},& \text{ if }i -j \geqslant n_{\mathrm{cr} }, \\
 \overline{\zeta} \big[ \overline{\gamma}_2(\varepsilon_\mathrm{m} ,\varepsilon_{\mathrm{p}}) \big]^{{i-j}}, &\text{ if }  0\leqslant i -j < n_{\mathrm{cr}},
  \end{cases}
\end{IEEEeqnarray}
where $\gamma_2(\varepsilon_\mathrm{m} )  \triangleq \sqrt{1- \alpha^{-1}_{\varepsilon_\mathrm{m}} \beta_\star }$, $\overline{\gamma}_2(\varepsilon_\mathrm{m},\varepsilon_{\mathrm{p}} )  \triangleq \sqrt{ 1 -  \alpha_{\varepsilon_\mathrm{m}}^{-1} \overline{\beta}(\varepsilon_{\mathrm{p}}) } $, $ \alpha_{\varepsilon_\mathrm{m}} \triangleq (1- 8\|P_\star\|^2\varepsilon_\mathrm{m}  )^{-1/2} $, and the positive constants $\zeta $ and $\overline{\zeta}$ are defined in \eqref{eq:PhiZeta} and \eqref{eq:PhiConstant2}.
\end{prop}

Proposition~\ref{prop:PhiBarAB} shows that if the modeling error is sufficiently small, the approximate state-transition matrix, formulated on any $(\hat{A},\hat{B}) \in \mathcal{S} (\varepsilon_\mathrm{m})$, is guaranteed to converge to zero. The two distinct decay rates in \eqref{eq:PhihatIneq} satisfy $\gamma_2 \leqslant \overline{\gamma}_2<1$, analogous to the two decay rates in Lemma~\ref{lem:BasicRate}. Proposition~\ref{prop:PhiBarAB} can also be of independent interest, e.g., for the stability analysis of finite-horizon LQR or the Kalman filter with modeling errors.

\begin{rem}
While the constants in the bounds \eqref{eq:PhiBarDecayRate} and \eqref{eq:PhihatIneq}, e.g., $\zeta$ and $\overline{\zeta}$, can be conservative, the rates $\gamma_1(\varepsilon_{\mathrm{m}})$, $\gamma_2(\varepsilon_\mathrm{m} )$, and $\overline{\gamma}_2(\varepsilon_\mathrm{m},\varepsilon_{\mathrm{p}})$ have the desired property that they recover the ideal rate $\sqrt{1-\beta_\star}$ of the true system in \eqref{eq:RateLstar} if $\varepsilon_\mathrm{m} =\varepsilon_\mathrm{p} =0$. 
\end{rem}

In the rest of this work, we will omit the dependence of $\gamma_1$, $\gamma_2$, and $\overline{\gamma}_2$ on $\varepsilon_\mathrm{m}$ and $ \varepsilon_{\mathrm{p}}$ for the simplicity of notation. With Lemmas~\ref{lem:Pdiff}, \ref{lem:PhiBar} and Proposition~\ref{prop:PhiBarAB}, we can now characterize $\varepsilon$ in \eqref{eq:RicciDiffPstar} as a function of $\varepsilon_\mathrm{p}$, $\varepsilon_\mathrm{m}$ and $N$ in the following result:  
 \begin{thm} \label{thm:RiccatiDiff}
Given $P$ with $\|P- P_\star\| \leqslant \varepsilon_\mathrm{p}$, if Assumptions~\ref{ass:ControlIndex}, \ref{ass:LQR}, \ref{ass:TechContPertur}, and \ref{ass:Dual} hold, then for $i \in \mathbb{Z}^+$, $\|\mathcal{R}^{(i)}_{\hat{A},\hat{B}}(P) - P_\star\|
\leqslant  \widehat{E}(\varepsilon_\mathrm{m}, \varepsilon_\mathrm{p},i)  \triangleq $
 \begin{align}\label{eq:PerturbRiccatiDiff}
 \hspace{-0.35cm} \Tilde{\zeta} \Bigg[ 
 \gamma_1^i \Tilde{\gamma}_2^i  \varepsilon_\mathrm{p}  +    \Tilde{\psi}(\varepsilon_\mathrm{m}) \Big (   \sum_{k=0}^{\min\{i, n_{\mathrm{cr}}\}-1} \overline{\gamma}_2^k\gamma_1^k+  \sum_{j=n_{\mathrm{cr}}}^{i-1} \gamma_2^j\gamma_1^j \Big) \Bigg],    
 \end{align}
 where $\Tilde{\psi}$ is defined in \eqref{eq:TidlePsiIneq} and satisfies $\Tilde{\psi} = O(\varepsilon_{\mathrm{m}})$, 
\begin{equation} \label{eq:rateThm}
\Tilde{\gamma}_2 = \begin{cases}\gamma_2, & \text{ if } i \geqslant n_{\mathrm{cr}}, \\
\overline{\gamma}_2, & \text{ otherwise}, 
    \end{cases}
\end{equation}
and $\Tilde{\zeta} = \sqrt{\beta_\star^{-1}}  \max\{ \overline{\zeta},\zeta \}$, with $\gamma_2$, $\overline{\gamma}_2$, $\zeta$, $\overline{\zeta}$ defined in Proposition~\ref{prop:PhiBarAB}, and $\gamma_1$ defined in \eqref{eq:Defgamma1Psi}.
\end{thm}

Note that $\lim_{i \to \infty}\widehat{E}(\varepsilon_\mathrm{m}, \varepsilon_\mathrm{p},i) $ exists iff $\varepsilon_{\mathrm{m}}$ is sufficiently small, as in Assumption~\ref{ass:TechContPertur}, such that $\gamma_1 \gamma_2 <1$, and thus the approximate Riccati iterations, formulated on any $(\hat{A},\hat{B}) \in \mathcal{S}  (\varepsilon_\mathrm{m})$, will converge. In this case, the upper bound \eqref{eq:PerturbRiccatiDiff} characterizes the convergence behavior of the Riccati iterations under the modeling error. As the iteration number $i$ increases, the term $ \gamma_1^i \Tilde{\gamma}_2^i  \varepsilon_\mathrm{p}$ in \eqref{eq:PerturbRiccatiDiff} decreases and drives $\mathcal{R}^{(i)}_{\hat{A},\hat{B}}(P)$ closer to $P_\star$; however, the remaining term increases and drives $\mathcal{R}^{(i)}_{\hat{A},\hat{B}}(P)$ away from $P_\star$, which reflects the propagation of the modeling error. The above two terms capture the trade-off between the prediction horizon and the modeling error in the convergence of the approximate Riccati iterations.

\subsection{Performance upper bound with a modeling error} \label{sec:PerError}
Based on \eqref{eq:Kmpc}, Theorem~\ref{thm:RiccatiDiff} and Lemma~\ref{lem:metalemma}, the tradeoff between the prediction horizon, the modeling error, and the terminal matrix error can be reflected in the performance of the RHC controller as follows.
\begin{lem} \label{prop:MPCperformance}
Given the nominal RHC controller \eqref{eq:MPCproblem}, $\widehat{E}$ in \eqref{eq:PerturbRiccatiDiff} and $c_\star$ in \eqref{eq:cstar}, if Assumptions~\ref{ass:ControlIndex}, \ref{ass:LQR}, \ref{ass:TechContPertur}, and \ref{ass:Dual} hold, and if $\varepsilon_\mathrm{m}+ \widehat{E}(\varepsilon_\mathrm{m}, \varepsilon_\mathrm{p},N-1) \leqslant 1/(40 \Upsilon_\star^4 \|P_\star\|^2 )$, then $\gamma_1 \gamma_2 < 1$ holds, and $A_\star - B_\star K_\mathrm{RHC}$ is Schur stable with $J_{K_\mathrm{RHC}}-J_{K_\star}   $
\begin{equation} \label{eq:MPCperforUnknown}
 \leqslant g(\varepsilon_\mathrm{m},\varepsilon_\mathrm{p},N) \triangleq c_\star  \big[\varepsilon_\mathrm{m}+ \widehat{E}(\varepsilon_\mathrm{m}, \varepsilon_\mathrm{p},N-1 ) \big]^2.
\end{equation}
\end{lem}

In Lemma~\ref{prop:MPCperformance}, $\widehat{E}(\varepsilon_\mathrm{m}, \varepsilon_\mathrm{p},N-1)$ depends on $N-1$ because of $\mathcal{R}^{ (N-1)}_{\hat{A},\hat{B}}(P)$ in \eqref{eq:Kmpc}. Moreover, the upper bound for $\varepsilon_\mathrm{m}+ \widehat{E}(\varepsilon_\mathrm{m}, \varepsilon_\mathrm{p},N-1)$ limits the joint effect of $\varepsilon_{\mathrm{m}}$, $\varepsilon_{\mathrm{p}}$, and $N$, such that the resulting nominal RHC controller can stabilize the unknown system with a suboptimality guarantee. When the modeling error is significantly large, the nominal controller may fail to stabilize the unknown system. In this case, a more accurate terminal value function can be considered, or a robust controller can be used to explicitly address the error\cite{mayne2014model}.

In addition, as $c_\star$ in \eqref{eq:cstar} increases as $\min\{n,m\}$ increases, a larger input dimension or state dimension may lead to a larger $c_\star$ and thus a larger performance gap $g$ in \eqref{eq:MPCperforUnknown}. Deriving the exact growth rate of $g$ when the state or input dimension increases is a subject of future work.

In the bound $g$, due to the tradeoff in the term $\widehat{E}$, increasing $N$ has a complex impact on the control performance due to the propagation of the modeling error. Then natural questions are whether an optimal prediction horizon $N_\star $ can be found to minimize the performance upper bound $g$, and how $N_\star$ varies when the errors $\varepsilon_\mathrm{m}$ and $\varepsilon_\mathrm{p}$ change.

\begin{exam} \label{exam:1}
To address the above questions, consider the case $i = N-1 < n_{\mathrm{cr}}$ as an example. Then, we can rewrite \eqref{eq:PerturbRiccatiDiff} into
\begin{align} \label{eq:Ereformu}
\widehat{E} = \Tilde{\zeta}  \Bigg [   \Big(\varepsilon_\mathrm{p} - \frac{\Tilde{\psi}(\varepsilon_\mathrm{m}) }{ 1- \gamma_1 \overline{\gamma}_2 }   \Big) \big(\gamma_1 \overline{\gamma}_2 \big)^{N-1}  +\frac{\Tilde{\psi}(\varepsilon_\mathrm{m}) }{ 1- \gamma_1 \overline{\gamma}_2 } \Bigg].
\end{align}
Note that in \eqref{eq:Ereformu}, the term
\begin{equation} \label{eq:CondiChangeN}
\varepsilon_\mathrm{p} - \frac{\Tilde{\psi}(\varepsilon_\mathrm{m}) }{ 1- \gamma_1\overline{\gamma}_2  }
\end{equation}
can be interpreted as a measure of the relative difference between $\varepsilon_\mathrm{p} $ and $\varepsilon_\mathrm{m}$, due to $\Tilde{\psi} = O(\varepsilon_{\mathrm{m}})$ and thus $\Tilde{\psi}(\varepsilon_\mathrm{m})/ ( 1- \gamma_1\overline{\gamma}_2 ) = O(\varepsilon_{\mathrm{m}})$. Then \eqref{eq:Ereformu} suggests that if the prediction horizon $N$ increases, the overall change of $\widehat{E}$ and $g$ depend on the sign of \eqref{eq:CondiChangeN}. More specifically, if $\varepsilon_{\mathrm{p}}$ is relatively smaller than $\varepsilon_\mathrm{m}$, i.e., \eqref{eq:CondiChangeN} is negative, then $\inf_N g(\varepsilon_\mathrm{m},\varepsilon_\mathrm{p},N)$ is attained at $N=1$; otherwise, $\inf_N g(\varepsilon_\mathrm{m},\varepsilon_\mathrm{p},N)$ is attained at the largest $N$ possible, i.e., $N =n_{\mathrm{cr}}$ in this example.
\end{exam}

The observation in Example~\ref{exam:1} can be generalized beyond the case $i< n_{\mathrm{cr}}$. To state the formal result, the following terms are relevant, and their relation can be established:
 \begin{equation}  \label{eq:proofRateCompare}
\hspace{-0.15cm} \varepsilon_\mathrm{p} - \frac{\Tilde{\psi}(\varepsilon_\mathrm{m}) }{ 1-\frac{ \gamma_1  \gamma_2}{ (\overline{\gamma}_2/\gamma_2)^{n_{\mathrm{cr}}-1 }     } }  \geqslant  \varepsilon_\mathrm{p} - \frac{\Tilde{\psi}(\varepsilon_\mathrm{m}) }{ 1- \gamma_1  \gamma_2 }  \geqslant  \varepsilon_\mathrm{p} - \frac{\Tilde{\psi}(\varepsilon_\mathrm{m}) }{ 1- \gamma_1  \overline{\gamma}_2 },
\end{equation}
which follows from $\gamma_2 \leqslant \overline{\gamma}_2 $. Similar to \eqref{eq:CondiChangeN}, the above terms represent three quantitative measures of the difference between the two errors $\varepsilon_\mathrm{m}$  and $\varepsilon_\mathrm{p} $. Their signs decide the optimal prediction horizon $N_\star$, as shown in the following result:
\begin{thm} \label{coro:ChangeN}
In the setting of Lemma~\ref{prop:MPCperformance}, let $\gamma_{\mathrm{ratio}} \triangleq  (\overline{\gamma}_2/\gamma_2)^{n_{\mathrm{cr}}-1}  \geqslant 1$, and consider the performance upper bound $g(\varepsilon_\mathrm{m},\varepsilon_\mathrm{p},N)$ in \eqref{eq:MPCperforUnknown}.
\begin{enumerate}[label=(\alph*)]
    \item If $\varepsilon_\mathrm{p} - \Tilde{\psi}(\varepsilon_\mathrm{m})/(1-\gamma_1 
\overline{\gamma}_2) >0 $, then $\inf_N g(\varepsilon_\mathrm{m},\varepsilon_\mathrm{p},N) = \lim_{N \to \infty} g(\varepsilon_\mathrm{m},\varepsilon_\mathrm{p},N)$ and $g(\varepsilon_\mathrm{m},\varepsilon_\mathrm{p},N)$ is strictly decreasing as $N$ increases;
    
    \item If $\varepsilon_\mathrm{p} - \Tilde{\psi}(\varepsilon_\mathrm{m})/[1-\gamma_1 \gamma_2   ] <0 $, then $\inf_N g(\varepsilon_\mathrm{m},\varepsilon_\mathrm{p},N) =  g(\varepsilon_\mathrm{m},\varepsilon_\mathrm{p},N_\star)$ with $N_\star = 1$, and 
   if $\varepsilon_\mathrm{p} - \Tilde{\psi}(\varepsilon_\mathrm{m})/[1-\gamma_1 \gamma_2/ \gamma_{\mathrm{ratio}}   ] <0 $ holds additionally, $g(\varepsilon_\mathrm{m},\varepsilon_\mathrm{p},N)$ is strictly increasing as $N$ increases;
     
    \item If $\varepsilon_\mathrm{p} - \Tilde{\psi}(\varepsilon_\mathrm{m})/(1-\gamma_1 \gamma_2 ) >0> \varepsilon_\mathrm{p} - \Tilde{\psi}(\varepsilon_\mathrm{m})/(1-\gamma_1  \overline{\gamma}_2)$, $g(\varepsilon_\mathrm{m},\varepsilon_\mathrm{p},N)$ is strictly increasing as $N \in \{1,\dots, n_{\mathrm{cr}}\}$ increases and strictly decreasing as $N\geqslant n_{\mathrm{cr}}+1$ increases. Moreover,
    \begin{itemize}
        \item $\inf_N g(\varepsilon_\mathrm{m},\varepsilon_\mathrm{p},N) =  g(\varepsilon_\mathrm{m},\varepsilon_\mathrm{p},N_\star)$ with $N_\star = 1$  if additionally 
        \begin{equation*} 
        \varepsilon_\mathrm{p} -  \Tilde{\psi}(\varepsilon_\mathrm{m}) \Big[ \frac{1- (\overline{\gamma}_2 \gamma_1)^{n_{\mathrm{cr}}}}{1- \gamma_1\overline{\gamma}_2} +\frac{(\gamma_1 \gamma_2)^{n_{\mathrm{cr}}}}{ 1- \gamma_1\gamma_2  }   \Big] \leqslant 0 ;        \end{equation*}

        \item $\inf_N g(\varepsilon_\mathrm{m},\varepsilon_\mathrm{p},N) = \lim_{N \to \infty} g(\varepsilon_\mathrm{m},\varepsilon_\mathrm{p},N)$ otherwise;
    \end{itemize}

    \item If $\varepsilon_\mathrm{p} - \Tilde{\psi}(\varepsilon_\mathrm{m})/[1-\gamma_1 \gamma_2   ] = \varepsilon_\mathrm{p} - \Tilde{\psi}(\varepsilon_\mathrm{m})/(1-\gamma_1 
\overline{\gamma}_2)=0$, $g(\varepsilon_\mathrm{m},\varepsilon_\mathrm{p},N)$ is a constant as $N$ varies,
\end{enumerate}
where $\gamma_2$ and $\overline{\gamma}_2$ are defined in Proposition~\ref{prop:PhiBarAB}, $\gamma_1$ is defined in \eqref{eq:Defgamma1Psi}, and $\Tilde{\psi}$ is defined in \eqref{eq:TidlePsiIneq} satisfying $\Tilde{\psi} = O(\varepsilon_{\mathrm{m}})$.
\end{thm}

Theorem~\ref{coro:ChangeN} shows in \eqref{eq:MPCperforUnknown}, the performance upper bound $g$ achieves its infimum either when $N=1$ or when $N\to \infty$, depending on the relative difference between $\varepsilon_{\mathrm{m}}$ and $\varepsilon_{\mathrm{p}}$, quantified by the terms in \eqref{eq:proofRateCompare}. More specifically, Theorem~\ref{coro:ChangeN}(a) shows that if $\varepsilon_{\mathrm{p}}$ is relatively larger than $\varepsilon_{\mathrm{m}}$ in the sense of \eqref{eq:CondiChangeN} being positive, $N\to \infty$ achieves the infimum of $g$. On the other hand, Theorem~\ref{coro:ChangeN}(b) shows that if $\varepsilon_\mathrm{m}$ is relatively larger than $\varepsilon_\mathrm{p}$, i.e., when $\varepsilon_\mathrm{p} - \Tilde{\psi}(\varepsilon_\mathrm{m})/[1-\gamma_1 \gamma_2   ] <0 $, $g$ is minimized by choosing $N=1$. Note that $g$ is exponential in $N$ based on \eqref{eq:PerturbRiccatiDiff} and \eqref{eq:Ereformu}. In addition, given the errors, $g$ is a monotonic function of $N$ in Theorem~\ref{coro:ChangeN}(a) and (b), but it shows a more complex behavior in (c).

\begin{rem}
Depending on the relative difference between $\varepsilon_{\mathrm{m}}$ and $\varepsilon_{\mathrm{p}}$, Theorem~\ref{coro:ChangeN} indicates that choosing $N=1$ or $N \to \infty$ could achieve a better control performance. However, even if $N \to \infty$ is desired, a finite $N$ can be sufficient to make $g$ close to its infimum, based on its exponential dependence on $N$. Considering additionally the faster decay rate in \eqref{eq:rateThm} when $i \geqslant n_\mathrm{cr}$, it can be beneficial to choose a finite $N$ satisfying $N > n_{\mathrm{cr}}$, where we recall that $n_{\mathrm{cr}}$ is the controllability index. This also suggests choosing $N=n_{\mathrm{cr} }$ or $n_{\mathrm{cr} }+1$ could be reasonable initial guesses, which can then be further tuned for performance improvement\footnote{In this work, the control horizon and the prediction horizon are equal.}. If $n_{\mathrm{cr}}$ is not known, the
state dimension can be an alternative, as a general controllable system
has $n_{\mathrm{cr}}=n$. Using $N=n$ is also suggested in \cite{Schutter2001model}, and in \cite{fruchard2012choice} for controlling nonholonomic vehicles.
\end{rem}

The above observation depends on the performance \textit{upper bound} in \eqref{eq:MPCperforUnknown} and thus may not reflect the actual behavior of $J_{K_\mathrm{RHC}}-J_{K_\star}$. However, the observation can indeed be seen in extensive simulations in Fig.~\ref{fig:ExamThm}. In this simulation study, $5$ random real systems $(A_\star,B_\star)$, with $n=10$ and $m=2$, are generated, where each entry in $A_\star$ and $B_\star$ is sampled from $\mathcal{U}_{[-2,2]}$ and $\mathcal{U}_{[0,1]}$, respectively. All systems have controllability index $n_{\mathrm{cr}}=5$. Moreover, each $(A_\star,B_\star)$ and the corresponding $P_\star$ are perturbed by random matrices\footnote{For $P_\star$, a random matrix is multiplied by its transpose to generate a positive semi-definite matrix perturbation.}, whose entries are sampled from $\mathcal{U}_{[0,0.001]}$. This is repeated for $5$ times for each real system, leading to $25$ approximate $(\hat{A},\hat{B},P)$ and $25$ nominal RHC controllers. For each controller, we vary its prediction horizon and compute the normalized performance gap as $[J_{K_\mathrm{RHC}}(N)-J_{K_\star}] /\max_N (J_{K_\mathrm{RHC}}(N)-J_{K_\star})$. 

The results are in Fig.~\ref{fig:ExamThm}. The performance gaps have initial transient behaviors, but converge quickly as $N$ increases. Most controllers indeed obtain their optimal control performance at $N=1$ or the largest $N$ until convergence. All of them converge after $N > n_{\mathrm{cr}}$, showing the potential benefit of a finite $N$ satisfying $N > n_{\mathrm{cr}}$ even if $N \to \infty$ is desired. This is important when the computational aspect is considered, as a larger $N$ leads to a higher computational cost.

However, there are a few cases where the optimal $N$ is finite but greater than $1$. Two examples are shown in Fig.~\ref{fig:Singular}, where $N=2$ and $N=3$ achieve the optimal performance, respectively. However, the performance difference between the optimal performance and the performance upon convergence is not significant. These examples show that the upper bound may not reflect the actual behavior of $J_{K_\mathrm{RHC}}-J_{K_\star}$.

\begin{figure}
    \centering
    \includegraphics[width=0.5\textwidth]{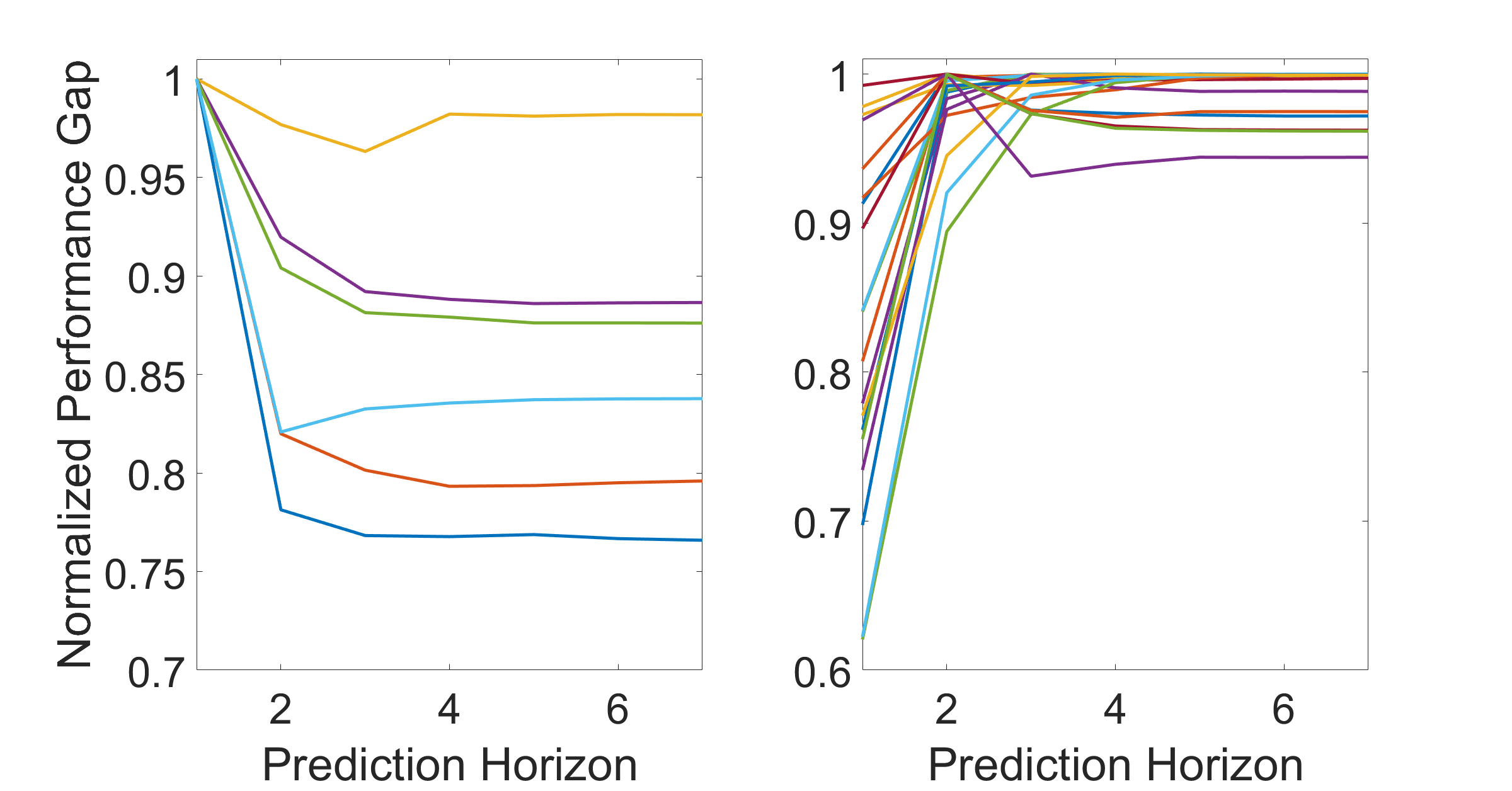}
    \caption{ Given $10$-dimensional random real systems, the performance gaps $J_{K_\mathrm{RHC}}-J_{K_\star}$ (after normalization) of $25$ nominal RHC controllers under varying $N$ are divided into two groups based on the trend of their performance and are shown in these two figures. }
    \label{fig:ExamThm}
\end{figure}

\begin{figure}
    \centering
    \includegraphics[width=0.4\textwidth]{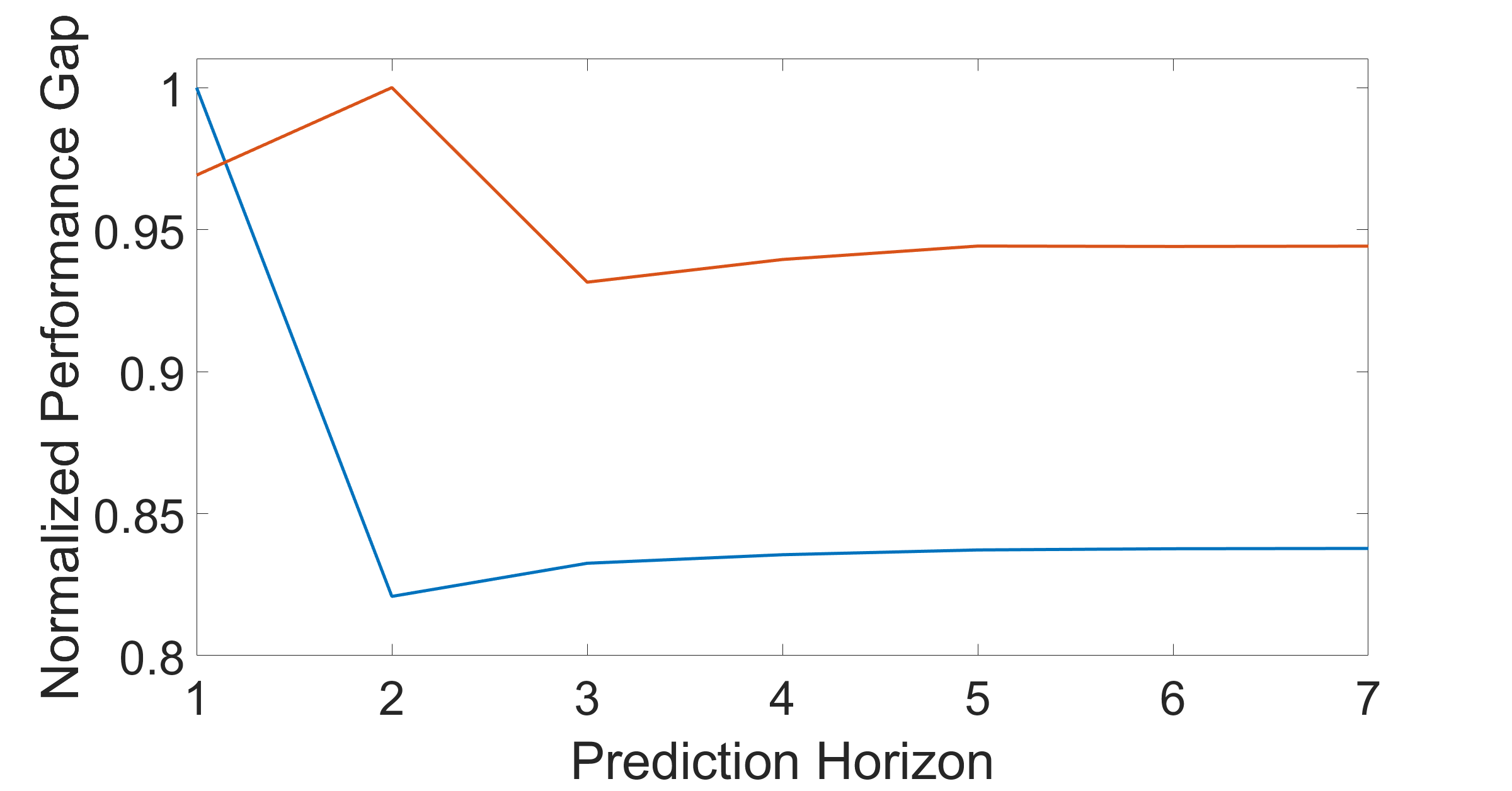}
    \caption{ Two nominal RHC controllers achieve their optimal closed-loop performance at $N=2$ and $N=3$, respectively. }
    \label{fig:Singular}
\end{figure}

\begin{rem}
    Lemma~\ref{prop:MPCperformance} is established under tight conditions, i.e., the bound $1/(40 \Upsilon_\star^4 \|P_\star\|^2 )$ and the bounds on $\varepsilon_{\mathrm{m}}$ in Assumption~\ref{ass:TechContPertur} can be small and thus may not be satisfied in practice. This is a common limitation of analytical error bounds \cite{mania2019certainty,simchowitz2020naive}. Despite this limitation, we can still obtain a non-trivial result in Theorem~\ref{coro:ChangeN}, and this result can be observed in simulations when going beyond the tight conditions, e.g., the modeling errors of the simulation in Fig.~\ref{fig:ExamThm} are actually greater than $1/(40 \Upsilon_\star^4 \|P_\star\|^2 )$. 
\end{rem}

We also consider the following practical case of Theorem~\ref{coro:ChangeN} where $P=0$ in the RHC controller \eqref{eq:MPCproblem}.

\begin{coro} \label{prop:MPCnoTerminal} 
Given the nominal RHC controller \eqref{eq:MPCproblem} with $P=0$, if Assumptions~\ref{ass:ControlIndex}, \ref{ass:LQR}, \ref{ass:TechContPertur}, and \ref{ass:Dual} hold, and if $\varepsilon_\mathrm{m}+ \widehat{E}(\varepsilon_\mathrm{m}, \|P_\star\|,N-1) \leqslant 1/(40 \Upsilon_\star^4 \|P_\star\|^2 )$, then $A_\star - B_\star K_\mathrm{RHC}$ is Schur stable with $J_{K_\mathrm{RHC}}-J_{K_\star} \leqslant g(\varepsilon_\mathrm{m},\varepsilon_\mathrm{p},N)$, where $\inf_N g(\varepsilon_\mathrm{m},\varepsilon_\mathrm{p},N) = \lim_{N \to \infty} g(\varepsilon_\mathrm{m},\varepsilon_\mathrm{p},N)$ and $g(\varepsilon_\mathrm{m},\varepsilon_\mathrm{p},N) $ is strictly decreasing as $N$ increases.
\end{coro}

The above result shows that for the RHC controller with zero terminal function, using a large prediction horizon, ideally the nominal LQR controller, is beneficial. This is because the other choice $N=1$ leads to $K_{\mathrm{RHC}}=0$, which will never stabilize the system if the system is open-loop unstable.

\section{Extensions to stabilizable systems} \label{sec:sta}
We generalize the previous results for controllable systems to stabilizable systems.
Without controllability to establish the fast decay rates $ 1 - \beta_\star  $ in \eqref{eq:KnownModel}  and $\gamma_2$ in \eqref{eq:PhihatIneq} for $i \geqslant n_{\mathrm{cr}}$, we only establish a single decay rate for all $i \geqslant 1$. These results for stabilizable systems are formalized in Appendix~\ref{appdix:sta}. The single decay rate greatly simplifies the analysis, which leads to a variant of Theorem~\ref{thm:RiccatiDiff} for stabilizable systems:
\begin{thm} \label{thm:stabi}
If Assumptions~\ref{ass:LQR}, \ref{ass:LQRgeneral2}, and \ref{ass:TechContPertur}\ref{ass4:i} hold, then for $i \in \mathbb{Z}^+$,  $ \|\mathcal{R}^{(i)}_{\hat{A},\hat{B}}(P) - P_\star\| $
\begin{equation} \label{eq:PerturbRiccatiDiffstab} 
\hspace{-0.22 cm} \leqslant  \widehat{E}_{\mathrm{sta}}(\varepsilon_\mathrm{m}, \varepsilon_\mathrm{p},i) \triangleq  \overline{\zeta} \Big[  
 \gamma_1^i  \overline{\gamma}_2^i \varepsilon_\mathrm{p}  +   \Tilde{\psi}(\varepsilon_\mathrm{m}) \Big(    \sum_{j=0}^{i-1} \overline{\gamma}_2^j \gamma_1^j  \Big) \Big],
 \end{equation}
where $\gamma_1$ is defined in \eqref{eq:Defgamma1Psi},  $\Tilde{\psi}$ is defined in \eqref{eq:TidlePsiIneq} and satisfies $\Tilde{\psi} = O(\varepsilon_{\mathrm{m}})$, and $\overline{\zeta}$, $\overline{\gamma}_2$  are defined in Proposition~\ref{prop:PhiBarAB}.
\end{thm}

Similarly to \eqref{eq:MPCperforUnknown}, \eqref{eq:PerturbRiccatiDiffstab} also leads to an error bound 
\begin{equation} \label{eq:ErrStabi}
  g_{\mathrm{sta}}(\varepsilon_\mathrm{m}, \varepsilon_\mathrm{p},N ) \triangleq c_\star [\varepsilon_{\mathrm{m}}+ \widehat{E}_{\mathrm{sta}}(\varepsilon_\mathrm{m}, \varepsilon_\mathrm{p},N-1 )]^2.  
\end{equation}
Equation~\eqref{eq:PerturbRiccatiDiffstab} can be rewritten as \eqref{eq:Ereformu}. Therefore, the observation from Theorem~\ref{coro:ChangeN} that the performance upper bound achieves its infimum at $N=1$ or $N\to \infty$ remains valid for stabilizble systems, now depending on the sign of \eqref{eq:CondiChangeN} only. 

By observing \eqref{eq:RiccatiDiffSameModel}, we note that for stabilizable systems, the faster rate $1 - \beta_\star$ in \eqref{eq:KnownModel} can also be established by letting $i$ be sufficiently large, e.g. larger than a constant $i_0$, such that the closed-loop matrices in $\Phi^{(0:i)}_{A,B}(P)$ become close to $L_\star$ in \eqref{eq:Lstar}. However, $i_0$ can be large and also does not have a clear interpretation, unlike the controllability index $n_{\mathrm{cr}}$ in \eqref{eq:KnownModel}. Therefore, this technical extension is not pursued here.

\section{Applications in learning-based control} \label{sec:Application}
Besides the obtained theoretical insights, our suboptimality analysis can also be used to analyze the performance of learning-based controllers. Results in this section generalize the results in \cite{mania2019certainty, dean2020sample,simchowitz2020naive} for learning-based LQR controllers with an infinite horizon to RHC controllers with an arbitrary prediction horizon. We will compare our performance upper bounds analytically with the above existing results, as is done in similar studies on regret analysis \cite{mania2019certainty,simchowitz2020naive}. 

\subsection{Offline identification and control}
In this subsection, we first obtain an estimate $(\hat{A},\hat{B})$ offline from measured data of the unknown real system \eqref{eq:TrueModel}, and then synthesize a controller \eqref{eq:MPCproblem} with zero terminal matrix $P=0$. This is the classical receding-horizon LQ controller \cite{bitmead1991riccati}. 

There are many recent studies on linear system identification and its finite-sample error bounds \cite{simchowitz2018learning,sarkar2019near,dean2020sample}. 
For interpretability, we consider a relatively simple estimator from \cite{dean2020sample}. Assume that the white noise $w_t$ is Gaussian, and we conduct $T$ independent experiments on the unknown real system \eqref{eq:TrueModel} by injecting independent Gaussian noise, i.e., $u_t \sim \mathcal{N}(0,\sigma_u^2 I)$ with $\sigma_u>0$, where each experiment starts from $x_0=0$ and lasts for $t_h$ time steps. This leads to the measured data $\{ (x_t^{(l)},u_t^{(l)})   \}_{t=0}^{t_h}$, $l = 1,\dots, T$, which is independent over the experiment index $l$. To avoid the dependence among the data for establishing the estimation error, we use one data sample from each independent experiment, and then an estimate $(\hat{A},\hat{B})$ is obtained from the least-squares (LS) estimator \cite{dean2020sample}:
\begin{equation} \label{eq:LSesti}
\begin{bmatrix}
\hat{A} & \hat{B}
\end{bmatrix} = \arg\min_{A,B} \sum_{l=1}^{T} \Big\|x_{t_h}^{(l)} - \begin{bmatrix}
A & B
\end{bmatrix} \begin{bmatrix}
x_{t_h -1}^{(l)} \\
u_{t_h -1}^{(l)}
\end{bmatrix} \Big\|^2.
\end{equation}
A bound $\varepsilon_\mathrm{m}$ can be obtained \cite[Prop.~1]{dean2020sample} so that with high probability and for some constant $c_{\mathrm{ls}}>0$,
\begin{equation} \label{eq:BoundLS}
\max\{\|\hat{A}-A_\star\|,\| \hat{B}-B_\star\| \}  \leqslant \varepsilon_\mathrm{m} =c_{\mathrm{ls}} / \sqrt{T}.\end{equation}

Consider the nominal controller \eqref{eq:MPCproblem} with $P=0$ and $(\hat{A},\hat{B})$ from \eqref{eq:LSesti}. We can provide an end-to-end control performance guarantee by combining the LS estimation error bound \eqref{eq:BoundLS} and our performance upper bound \eqref{eq:ErrStabi}. We consider \eqref{eq:ErrStabi} here for generality as it needs more relaxed assumptions than \eqref{eq:MPCperforUnknown}. To this end, we first simplify \eqref{eq:ErrStabi} for interpretability.
\begin{prop} \label{lem:SimplePerfor}
Given the controller \eqref{eq:MPCproblem} with $P=0$, if Assumptions~\ref{ass:LQR}, \ref{ass:LQRgeneral2}, and \ref{ass:TechContPertur}\ref{ass4:i} hold, and if $ \varepsilon_\mathrm{m}+ \widehat{E}_{\mathrm{sta}}( \varepsilon_\mathrm{m}, \|P_\star\|,N-1) \leqslant 1/(40 \Upsilon_\star^4 \|P_\star\|^2 )$, then $A_\star - B_\star K_\mathrm{RHC}$ is Schur stable with 
\begin{equation*}
J_{K_\mathrm{RHC}}-J_{K_\star} \leqslant C \left(\mu_\star^{N-1} +\varepsilon_\mathrm{m}   \right)^2,
\end{equation*}
for some $C >0$ that is independent of $N$ and $\varepsilon_\mathrm{m}$, where $\mu_\star \triangleq \sqrt{1-  \beta_\star /  (2)^{3/2}} \in (0,1)$.
\end{prop}

Combining Proposition~\ref{lem:SimplePerfor} with the estimation error bound \eqref{eq:BoundLS} from \cite{dean2020sample} directly leads to the following end-to-end performance guarantee:
\begin{coro} \label{coro:OffSI}
Consider the estimate $(\hat{A},\hat{B})$ from the estimator \eqref{eq:LSesti} and the controller \eqref{eq:MPCproblem} with $P=0$. If Assumptions~\ref{ass:ControlIndex}, \ref{ass:LQR} hold, and if the noise $w_t$ is Gaussian and $c_{\mathrm{ls}} / \sqrt{T}+ \widehat{E}_{\mathrm{sta}}(c_{\mathrm{ls}} / \sqrt{T}, \|P_\star\|,N-1) \leqslant 1/(40 \Upsilon_\star^4 \|P_\star\|^2 )$, then for any $\delta\in (0,1)$, there exists a sufficiently large $T$ such that with probability at least $1-\delta$, $A_\star - B_\star K_\mathrm{RHC}$ is Schur stable with 
\begin{equation*}
J_{K_\mathrm{RHC}}-J_{K_\star} \leqslant  C \left( \mu_\star^{N-1} + \sqrt{ \log(1/\delta)/ T }   \right)^2,\end{equation*}
for some $C >0$ that is independent of $N$ and $T$.
\end{coro}

Corollary~\ref{coro:OffSI} reveals the effect of the prediction horizon and the number of data samples on the control performance. The performance gap is $O(1/T)$ if $N \to \infty$, i.e., when the nominal LQR controller is considered. This growth rate agrees with the recent study \cite{mania2019certainty} of the LQR controller. When $N$ is finite, an additional error $\mu_\star^{N-1}$ exists and decreases as $N$ increases. When $T\to \infty$, the estimated model converges to the real system with high probability, leading to performance gap $O(\mu_\star^{N-1})$, i.e., the performance of the controller converges exponentially to the optimal one as $N$ increases, which matches our observation in \eqref{eq:JgapKnown}. Note that controllability in Assumption~\ref{ass:ControlIndex} is introduced in Corollary~\ref{coro:OffSI}, required by the estimation error bound \cite{dean2020sample}. Moreover, Assumption~\ref{ass:TechContPertur}\ref{ass4:i} is absent, as it is satisfied by a sufficiently large $T$. 
\begin{rem}
We have used data from multiple independent experiments for estimating the model; however, in practice, maybe only one experiment can be conducted. In this case, an LS estimator with the data from a single experiment can be used, with its error bound studied in \cite{simchowitz2018learning,sarkar2019near}.
\end{rem}

\subsection{Online learning control with regret bound} \label{sec:OnlineRegret}
In this subsection, we consider the adaptive LQR controller algorithm in \cite{simchowitz2020naive}, which has a state-of-the-art regret guarantee with greedy exploration. While there are other exploration strategies \cite{tsiamis2023statistical}, the more fundamental greedy exploration is chosen to demonstrate the application of our analysis.

The nominal LQR controller within the adaptive control scheme in \cite{simchowitz2020naive} is replaced by the receding-horizon LQ controller \eqref{eq:MPCproblem} with $P=0$. Other cost matrices $Q$ and $R$ are fixed throughout the close-loop operation. To match the setting in \cite{simchowitz2020naive}, in this subsection, we assume that $w_t$ is Gaussian with $\sigma_w=1$ and the initial condition is $x_0=0$. Following \cite{simchowitz2020naive}, we assume that a stabilizing but possibly suboptimal controller gain $K_0$ for the unknown true system is given a priori.

We briefly introduce the main idea of \cite[Algorithm 1]{simchowitz2020naive}, and the details can be found in \cite{simchowitz2020naive}. Starting from $x_0=0$ and a random input $u_0 \sim \mathcal{N}(0,I)$, the algorithm utilizes a fixed controller gain $K_k$ within each time step period $ [t_k,t_{k+1})$, where $t_k = 2^{k-1}$ and $k \in \{1,2,\dots\}$ is the period index. In the first few periods, the gain $K_0$ with a Gaussian perturbation is used, i.e., $u_t = -K_k x_t +g_t$ with $K_k = K_0 $ and $g_t \sim \mathcal{N}(0,I)$, to stabilize the unknown system. The perturbation $g_t$ ensures the informativity of the data for estimating the system later.

After collecting sufficient data at time step $t_{k_0}$ for some period $k_0$, the controller conducts the following steps for every period $k \geqslant k_0$. A new model $(\hat{A}_k,\hat{B}_k)$ is re-estimated at the initial time step $t_k$ of period $k$. It is obtained from the LS estimator using the data $\{z_t\}_{t=t_{k-1}}^{t_k-1}$ from the last period $k-1$, where $z_t \triangleq \begin{bmatrix}
    x_t^\top & u_t^\top
\end{bmatrix}^\top$, as 
\begin{equation} \label{eq:LSAdaptive}
\begin{bmatrix}
\hat{A}_k & \hat{B}_k
\end{bmatrix} = \arg\min_{A,B} \sum_{t=t_{k-1}}^{t_k-1} \Big\|x_{t+1} - \begin{bmatrix}
A & B
\end{bmatrix} z_t\Big\|^2.
\end{equation}
Then, for $t \in [t_k,t_{k+1})$, the algorithm employs a new RHC controller gain $K_{\mathrm{RHC},k}$, formulated on\footnote{In \cite{simchowitz2020naive}, the estimated model, if not accurate, is modified via projection. These details are presented in \cite{simchowitz2020naive}. This projection step is redundant if the estimated model is sufficiently accurate. Moreover, it relies on a conservative error bound, leading to practical issues as discussed in Section~\ref{sec:SimuAdaptiveRHC}. Therefore, the projection step is ignored later in the simulations of Section~\ref{sec:SimuAdaptiveRHC}. } $(\hat{A}_k,\hat{B}_k)$ via either the implicit form \eqref{eq:MPCproblem} or the explicit form \eqref{eq:Kmpc}, with a random perturbation, i.e., $u_t = -K_{\mathrm{RHC},k} x_t + \sigma_k g_t$, where $\sigma_k>0$ determines the perturbation size and decays as $k$ increases. 

In this case, let $T$ denote the total number of time steps that have passed, and let $k_{T}$ denote the final period index. Let $\mathcal{C}: \mathbb{R}^n \to \mathbb{R}^m$ denote the above adaptive controller, and its performance is typically characterized by \textit{regret} \cite{simchowitz2018learning,lale2022reinforcement}:
$$
\mathrm{Regret}(T) \triangleq \sum_{t=0}^T \left[ x_t^\top Q x_t +  \mathcal{C}^\top(x_t) R \mathcal{C}(x_t) - J_{K_\star} \right],
$$
which measures the accumulative error of a particular realization under the controller. It is ideal to have a sublinear regret such that as $T \to \infty$, $\mathrm{Regret}(T)/T$ converges to zero with high probability, i.e., the average performance of the adaptive controller is optimal.

As shown in \cite[Sect. 5.1]{simchowitz2020naive} and \cite[Appendix G.2]{simchowitz2020naive}, the following holds with high probability for the controller:
$$
\mathrm{Regret}(T) = \Tilde{O} \left(   \sum_{k=k_0}^{k_{T}} t_k ( J_{K_{\mathrm{RHC},k}} -J_{K_\star} ) + \sqrt{T}   \right).
$$
Combining the above equation with Proposition~\ref{lem:SimplePerfor} leads to
\begin{equation} \label{eq:RegretBound}
\kern-0.3em \mathrm{Regret}(T) = \Tilde{O}\left(  \sum_{k=k_0}^{k_{T}} t_k \big[ \mu_\star^{N-1} + 1/(t_k ^{1/4})   \big]^2+ \sqrt{T}   \right), 
\end{equation}
where we use the fact that the estimated model at time step $t_k$ has error $\varepsilon_\mathrm{m} = O( 1/t_k ^{1/4})  $ with high probability \cite[Lem. 5.4]{simchowitz2020naive}.

The regret upper bound \eqref{eq:RegretBound} provides some interesting information. Firstly, note that $\sum_{k=k_0}^{k_{T}} \sqrt{t_k} \leqslant \bar{C}\sqrt{T}$ for some constant $\bar{C}>0$. Therefore, if we have an infinite prediction horizon $N \to \infty$, i.e., when the nominal LQR is considered, then $\mathrm{Regret}(T) = \Tilde{O}(\sqrt{T})$, which matches the rate of the adaptive LQR controllers in \cite{mania2019certainty,simchowitz2020naive}. 

If we have a finite $N$, \eqref{eq:RegretBound} shows with high probability, \begin{equation} \label{eq:FinalRegret}
    \mathrm{Regret}(T) =\Tilde{O}\left( T \mu_\star^{N}+\sqrt{T} \right), \end{equation}
where the regret is linear in $T$. This observation matches the result in \cite{yu2020power}, where the regret of a linear unconstrained RHC controller, with a fixed prediction horizon and an exact system model, is linear in $T$. This linear regret is caused by the fact that even if the model is perfectly identified, the RHC controller still deviates from the optimal LQR controller due to its finite prediction horizon.

To achieve a sublinear regret, \eqref{eq:FinalRegret} suggests that an adaptive prediction horizon is preferred. As also suggested in \cite{yu2020power} but for the case of a known system, we can update $N$ via $$N = O(\log(t_k)),$$ e.g., updating $N= \lfloor -\log(t_k)/(4 \log(\mu_\star) ) \rfloor$, when the model is re-estimated, which again leads to a sublinear regret $\Tilde{O}(\sqrt{T}) $ according to \eqref{eq:RegretBound}. Note that this is the optimal rate for the regret \cite{simchowitz2020naive}. The intuition of this choice is that, with a more refined model due to re-estimation, the prediction horizon can be increased adaptively to improve the control performance.
\subsection{Simulation of adaptive RHC} \label{sec:SimuAdaptiveRHC}

We use simulations to demonstrate the main theoretical insights from Section~\ref{sec:OnlineRegret}: For the adaptive RHC algorithm, while a fixed prediction horizon leads to a linear regret growth rate $\Tilde{O}\left( T\right)$, an adaptive horizon, being a logarithmic function of time, leads to a more desired sublinear rate $\Tilde{O}(\sqrt{T})$.

While the algorithm based on \cite{simchowitz2020naive} in Section~\ref{sec:OnlineRegret} has a regret guarantee, a conservative modeling error bound $\varepsilon_{\mathrm{m}}$, analogous to \eqref{eq:BoundLS}, is utilized in \cite{simchowitz2020naive}. This conservatism causes practical issues, e.g., despite the actual modeling error being small, the if condition $\varepsilon_{\mathrm{m}}< \bar{\delta}(\hat{A},\hat{B})$ for some $\bar{\delta}$ in the algorithm is hardly met as $\varepsilon_{\mathrm{m}}$ is too large. Therefore, by introducing tuning parameters, we slightly modify the algorithm in Section~\ref{sec:OnlineRegret} to avoid analytical error bounds. The resulting algorithm is more practical and suffices to demonstrate the theoretical results and the corresponding insights. We summarize the modifications here:
\begin{enumerate}[label=(\roman*)]
    \item The key period index $k_0$ in Section~\ref{sec:OnlineRegret} is chosen to be the first period $k$ such that $\sum_{t=0}^{t_k-1} z_tz_t^\top \succeq \delta_1 I$ holds, where $\delta_1 \in [1, \infty)$ is a tuning parameter.

    \item The LS estimator \eqref{eq:LSAdaptive} uses data $\{z_t\}_{t=0}^{t_k-1}$ from all the previous periods, instead of only the last period.

    \item We let $\sigma_k = \min\{1, \big(\delta_2 / \sqrt{t_k}\big)^{1/2}\}$, where $\delta_2 \in (0,\infty)$ is a tuning parameter instead of being computed via an analytical formula as in \cite{simchowitz2020naive}.
\end{enumerate}
In the above, (i) ensures that the collected data is informative to obtain an accurate LS estimate of the model. Point (ii) ensures the updated estimate is not worse than the previous one, as more data is used for estimation. With point (iii), the exploration effort decays as the period index increases.

We use the modified algorithm to control the unknown system with the following true system matrices:
$$
A_\star = \begin{bmatrix}
 0.5 & -2 & 0.9\\
 0.6 & 0.1 & 1.8\\
 1.3& -0.3 & 1.6
\end{bmatrix}, \text{ }B_\star = \begin{bmatrix}
 0.3  & 1\\
 0.9 & 1 \\
 0& 0.9 
\end{bmatrix}.
$$
Starting from $x_0=0$, we compare the regret growth of the adaptive RHC controller with a fixed prediction horizon $N=3$ and the controller with an adaptive prediction horizon $N = \max\{3, \lfloor \log(t_k)\rfloor\}$, where $N$ is updated when the controller is updated.  Other parameters are $\sigma_w=1$, $Q=I$, $R=I$, $P=0$, $\delta_1 = 4$, $\delta_2 = 0.5$, and $K_0$ is the LQR controller gain for stabilization. Note that although the true system and its LQR controller are unknown, we choose the LQR controller to be the initial controller for simplicity. The choice of $K_0$, if stabilizing, does not affect our illustration of the regret growth.

As the regret is a random variable due to disturbance $w_t$, we conduct $100$ simulations for each controller. The regret growth over time, normalized by $\sqrt{t}$ or $t$, is shown in Fig.~\ref{fig:regret}. The regret grows linearly in $T$ under the fixed $N$ but in the order of $\sqrt{T}$ under the adaptive $N$. This illustrates the advantage of the adaptive $N$ for regret minimization. When the nominal LQR is used within the adaptive algorithm, the regret remains sublinear with a larger standard deviation in the initial period.

While the modified algorithm does not preserve the theoretical regret guarantee in Section~\ref{sec:OnlineRegret}, it remains useful to demonstrate the impact of the adaptive prediction horizon on regret growth. Moreover, this algorithm has a structure similar to that of existing algorithms with regret guarantees, e.g., the adaptive LQ Gaussian (LQG) controller in \cite{athrey2024regret}. Our future work will investigate the possibility of extending existing analyses to the modified algorithm. Note that the computational increase with a larger $N$ is minor in the current setting and case study, as no online optimization is needed. A larger $N$ only requires more Riccati iterations for computing the controller. For example, in Matlab $2024$a with an Intel Core i9-13950HX CPU, the Riccati iteration takes on average around $3 \cdot 10^{-6}$ seconds for $N=3$, $1.8 \cdot 10^{-5}$ seconds for $N=15$, and solving the Riccati equation for LQR takes $2.4 \cdot 10^{-4}$ seconds. Therefore, the computational cost is a more important problem for more general settings with constraints or nonlinear systems, requiring online optimization. 

Changing the horizon adaptively is also investigated in \cite{krener2018adaptive,bohn2023optimization} and the references therein, which are either heuristic or not analyzed in the context of regret. This idea also has implications for, e.g., robot control in a partially unknown environment \cite{chen2022real}. There, as more data is collected, the RHC/MPC controller of the robot can increase its prediction horizon to improve performance. The potential challenges are the increase in computation time and handling data that has poor quality. In the latter case, the model may not improve over time, and thus it is important to decide when to increase the prediction horizon. These challenges will be investigated in future work. 

 \begin{figure}[t]
 \hspace{-0.3cm}
 \begin{minipage}{0.24\textwidth}
 \centering
 \includegraphics[width=1\textwidth]{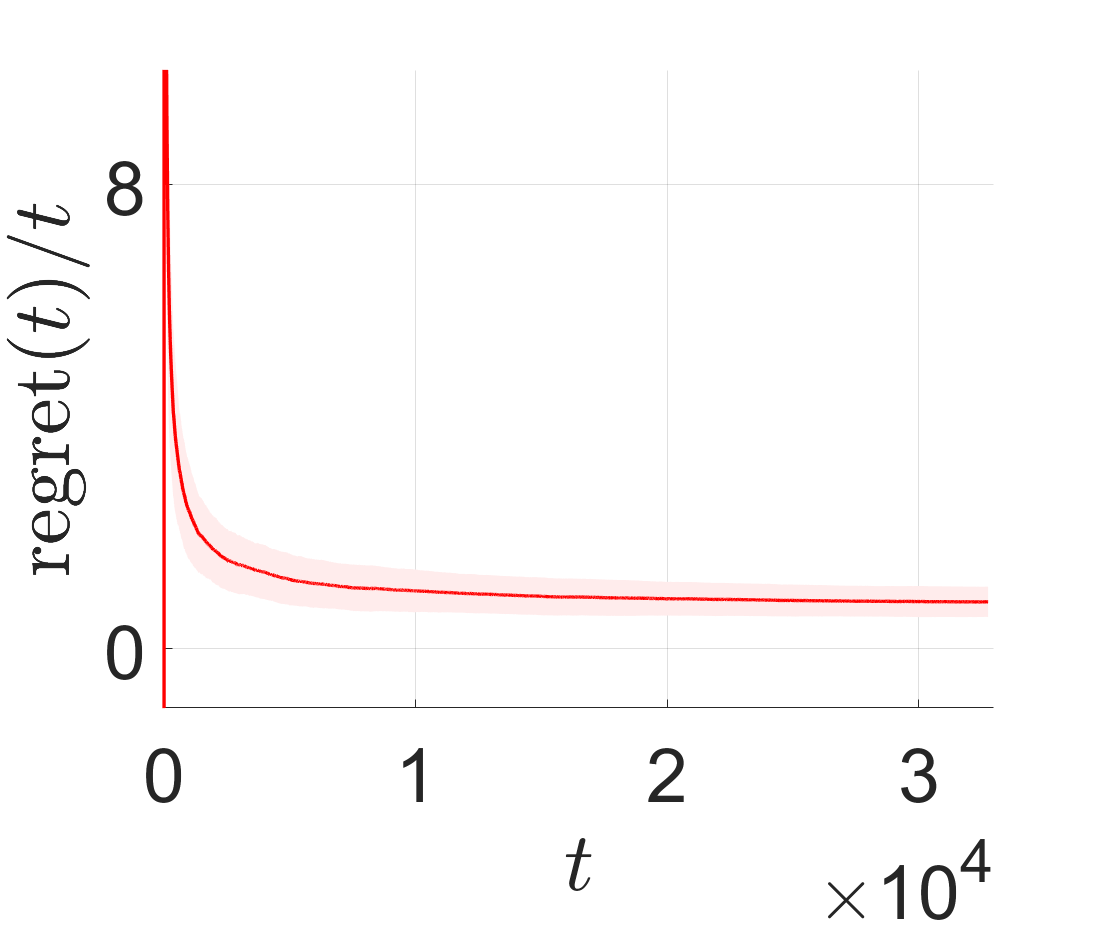}
 \end{minipage}
\begin{minipage}{0.24\textwidth}
 \centering
 \includegraphics[width=1.15 \textwidth]{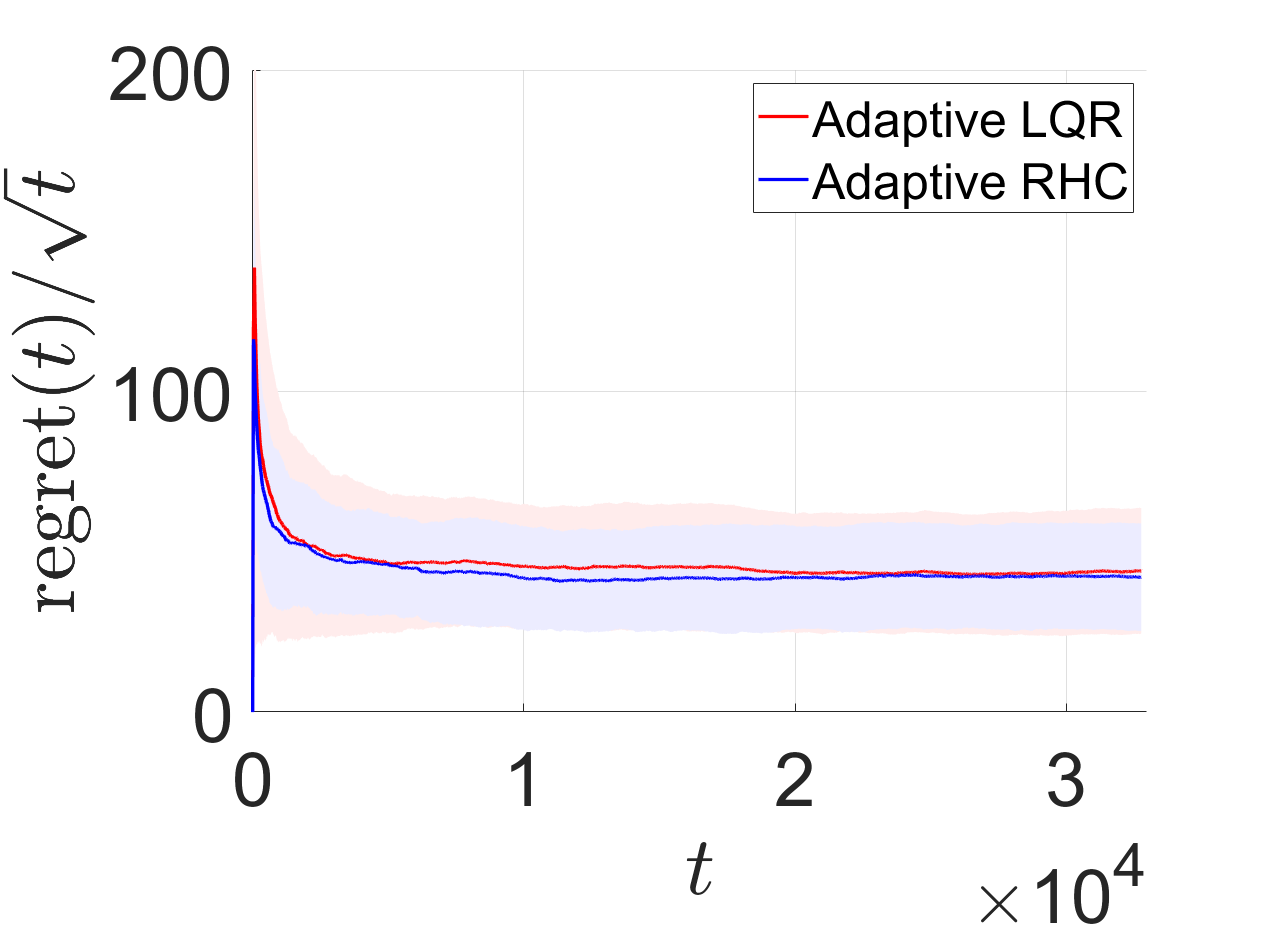}
 \end{minipage}
 \caption{For the RHC with a fixed $N$ (left) and the one with an adaptive $N$ (right), the mean and standard deviation (shaded region) of the normalized regrets, $\mathrm{regret}(t)/t$ for the fixed $N$ and $\mathrm{regret}(t)/\sqrt{t}$ for the adaptive $N$, over $100$ simulations are shown. After initial transient behavior, the normalized regrets converge to positive constants. The normalized regret of the adaptive LQR (right) is also shown.}
 \label{fig:regret}
 \end{figure}

\section{Conclusions}
This work analyzes the suboptimality of RHC under the joint effect of the modeling error, the terminal value function error, and the prediction horizon in the LQ setting.
 By deriving a novel perturbation analysis of the Riccati difference equation, we have obtained a novel performance upper bound of the controller. The bound suggests that letting the prediction horizon be $1$ or $+\infty$ can potentially be beneficial, depending on the relative difference between the modeling error and the terminal matrix error. Moreover, when an infinite horizon is desired, a prediction horizon larger than the controllability index can be sufficient for achieving a near-optimal performance. Besides the above insight, this obtained performance upper bound has also been shown to be useful for analyzing the performance of learning-based receding-horizon LQ controllers. 

Extending the results to more general settings with constraint and nonlinear systems is an important direction. The overall steps of our analysis can be applied to nonlinear systems; however, the technical derivations differ significantly. For example, instead of the Riccati iteration, the value iteration needs to be considered. This generalization can be found in the follow-up work \cite{liu2024certainty}. Future work also includes the derivation of tighter performance upper bounds to capture the non-trivial optimal horizon, which is finite but larger than $1$, and enhancing the adaptive RHC algorithm to achieve both practicality and theoretical guarantees. It is also important to benchmark the algorithm against other state-of-the-art methods.

\appendices
\section{Technical tools}
Technical tools from the literature are collected here.
\begin{lem} \label{lem:inequalProof}
 Given $0< a\leqslant b<1$ and any positive integer $i$, we have $(1-b^{i+1})/(1-a^{i+1}) - (1-b^{i})/(1-a^{i}) \geqslant 0.$
\end{lem}
\begin{proof}
This result is obtained by applying formula $a^i - b^i = (a-b)(a^{i-1}+a^{i-2}b+\dots +ab^{i-2}+b^{i-1})$.
\end{proof}

\begin{lem}{(\cite[Lem. 7]{mania2019certainty})} \label{lem:tool1Mania}
Given $M \in \mathbb{S}^n_+$ and $N \in \mathbb{S}^n_+$, it holds that $\|N(I+MN)^{-1} \| \leqslant \|N\|$. 
\end{lem}

\begin{lem}{(\cite[Lem. 5]{mania2019certainty})} \label{lem:PerturbRate}
Consider $M \in \mathbb{R}^{n \times n}$, $c \geqslant 1$, $\rho \geqslant \rho(M)$ such that $\| M^i\|\leqslant c \rho^i$ for any $i \in \mathbb{Z}^+$. For any $i \in \mathbb{Z}^+$ and $\Delta \in \mathbb{R}^{n \times n}$,
$
\|(M+\Delta)^i\| \leqslant c (c\|\Delta\|+\rho)^i$ holds.
\end{lem}

Important identities from \cite{bitmead1991riccati}, \cite[eq. (2.3) and (4.2)]{del2021note} are collected in the following result:
\begin{lem} \label{lem:RicaDiffEqu}
If $R \succ 0$, for any $P$, $P_1$, $P_2 \in \mathbb{S}_+^n$ and $i \in \mathbb{Z}^+$, 
\begin{enumerate}[label=(\alph*)] 
    \item $\mathcal{R}_{A,B}(P_1) - \mathcal{R}_{A,B}(P_2)=[A-B \mathcal{K}_{A,B}(P_1) ]^\top (P_1-P_2) [A-B \mathcal{K}_{A,B}(P_2) ] ;
$ \label{eq:IdentiRicatiDiif} 
\item $\mathcal{R}_{A,B}(P)=[A -B  \mathcal{K}_{A,B}(P) ]^\top P  [A -B \mathcal{K}_{A,B}(P)]+ \mathcal{K}^\top_{A,B}(P) R \mathcal{K}_{A,B}(P)+ Q$; \label{eq:IdentiRicatiLyapu}

\item  $\mathcal{F}(P) \triangleq   B ( R+ B^\top P  B)^{-1} B^\top = S_B (I+ P S_B)^{-1} $;

\item $\Phi^{(0:i)}_{A,B}(P_1) =\big[ I+ \mathcal{O}_{A ,B  }^{(i)}(P_1) (P_2 -P_1) \big] \Phi^{(0:i)}_{A,B}(P_2)$ with $ \mathcal{O}_{A ,B  }^{(i)}$ defined in \eqref{eq:DefOmatrix}. \label{eq:IdentiPhi}
\end{enumerate}\end{lem}

The following perturbation analysis of the Riccati equation from \cite[Prop.~6]{simchowitz2020naive} and \cite[Lem.~B.5, B.8]{simchowitz2020naive} is utilized.
\begin{lem}{(\cite{simchowitz2020naive})} \label{lem:PerturbationSimchow}
Given $Q$, $R$, and $(A_\star,B_\star)$, consider any alternative system $(\hat{A},\hat{B})$ satisfying $\max\{\|\hat{A} - A_\star\|, \| \hat{B} - B_\star\| \} \leqslant \varepsilon_\mathrm{m}$. If Assumptions~\ref{ass:LQR} and \ref{ass:LQRgeneral2} hold, then
\begin{enumerate}
    \item $P_\star \succeq I$, $\| P_\star\| \geqslant \| L_\star\|^2$, and $\| P_\star\| \geqslant \| K_\star\|^2$;
    
    \item if $\varepsilon_\mathrm{m} < 1/(8 \|P_\star\|^2 )$, then $(\hat{A},\hat{B})$ is stabilizable and 
    \begin{equation} \label{eq:PerturP}
\kern-0.3em   \| \hat{P} \| \leqslant \alpha_{\varepsilon_\mathrm{m}} \| P_\star\|, \quad  \|\hat{K} - K_\star\| \leqslant 7 \alpha_{\varepsilon_\mathrm{m}}^{7/2} \|P_\star\|^{7/2} \varepsilon_\mathrm{m}  \end{equation}
 where $ \alpha_{\varepsilon_\mathrm{m}} = (1- 8\|P_\star\|^2\varepsilon_\mathrm{m}  )^{-1/2}$, $\hat{P}$ is the fixed point of the Riccati equation and $\hat{K}$ is the LQR control gain for $(\hat{A},\hat{B})$.
\end{enumerate}
\end{lem}

Another tool is the stability analysis of time-varying systems, directly implied by the proof of \cite[Thm. 23.3]{rugh1996linear}.
\begin{lem} \label{lem:Stability}
Consider a system $\bar{x}_{k+1} = A_k \bar{x}_k$ with $A_k \in \mathbb{R}^{n \times n}$ and $\phi^{(j:i)} \triangleq A_{i-1}\dots A_{j} $ for integers $i \geqslant j \geqslant 0$. If there exists a matrix sequence $W_k \in \mathbb{S}^n$ satisfying for $k\in [j,i]$
\begin{align*}
   b_{\mathrm{l}}I \preceq W_k \preceq b_{\mathrm{u}}I, \text{ and }
    A_k^\top W_{k+1} A_k -   W_{k}  \preceq  - aI, 
\end{align*}
for some positive constants $b_{\mathrm{l}}$, $b_{\mathrm{u}}$, and $a$, then $\| \phi^{(j:i)} \| \leqslant  \sqrt{b_{\mathrm{u}}/b_{\mathrm{l}}   } \big(\sqrt{1 -  a/b_{\mathrm{u}} }\big)^{i-j} $ with $a/b_{\mathrm{u}} \in (0,1]$.
\end{lem}

\section{}

\subsection{Proof of Lemma~\ref{lem:DecayL}}
Lemma~\ref{lem:RicaDiffEqu}\ref{eq:IdentiRicatiLyapu} shows $- \underline{\sigma}(Q)I \succeq  L_\star^\top P_\star L_\star - P_\star$, and in addition, we have $ \underline{\sigma}(Q)I\preceq \underline{\sigma}(P_\star) I \preceq P_\star \preceq \overline{\sigma}(P_{\star} )I$. Then the direct application of Lemma~\ref{lem:Stability} proves the result. Note that $\beta_\star<1$ holds; otherwise, $P_\star=Q$ and thus $A=0$ based on Lemma~\ref{lem:RicaDiffEqu}\ref{eq:IdentiRicatiLyapu}, which contradicts with Assumption~\ref{ass:LQR}.

\subsection{Proof of Lemma~\ref{lem:BasicRate}}
We first define the following Gramian matrix: for $i \geqslant 1$
\begin{equation}  \label{eq:DefOmatrix}
\mathcal{O}_{A,B  }^{(i)}(P)  \triangleq \sum_{k=0}^{i-1} \mathcal{L}^k_{A,B}( P^{(k)}   ) \mathcal{F}\big( P^{(k)}    \big)  [\mathcal{L}^k_{A,B}(P^{(k)}  ) ]^\top,
\end{equation}
where $\mathcal{F}(P) \triangleq B ( R+ B^\top P  B)^{-1} B^\top$, $\mathcal{O}_{A,B  }^{(0)}(P)\triangleq 0 $, and recall that $ P^{(i)}$ is a shorthand notation for $ \mathcal{R}^{(i)}_{A,B}(P)$.

The proof for the case $i \geqslant n_{\mathrm{cr}}$ is achieved based on Lemma~\ref{lem:BasicRate} and \cite[Lem. 4.1]{del2021note}. We highlight the key steps for completeness. If $i \geqslant n_{\mathrm{cr}}$, matrix $\mathcal{O}_{A_\star ,B_\star  }^{(i)}(P_\star)$ has full rank due to Assumption~\ref{ass:ControlIndex}. Lemma~\ref{lem:RicaDiffEqu}\ref{eq:IdentiPhi} shows $\big[ I+ \mathcal{O}_{A_\star ,B_\star  }^{(i)}(P_\star) (P- P_\star) \big] 
 \Phi^{(0:i)}_{A_\star,B_\star}(P)=  L_\star^i $ for $i \in \mathbb{Z}^+$. Then based on Assumption~\ref{ass:LQR} and \cite[Lem. 4.1]{del2021note}, $  I+ \mathcal{O}_{A_\star ,B_\star  }^{(i)}(P_\star) (P- P_\star) $ is non-singular for $i \geqslant n_{\mathrm{cr}}$, and moreover, its inverse is uniformly upper bounded for any $i \geqslant n_{\mathrm{cr}}$ and any $P \succeq 0$. Combining \eqref{eq:RiccatiDiffSameModel} and the above shows for $i \geqslant n_{\mathrm{cr}}$, $\mathcal{R}^{(i)}_{A_\star,B_\star}(P) - P_\star=
 (L_\star^i)^\top \big[ I+ \mathcal{O}_{A_\star ,B_\star  }^{(i)}(P_\star) (P- P_\star) \big]^{-\top} (P-P_\star) L_\star^i, $
 which together with \eqref{eq:RateLstar} concludes the proof for $i \geqslant n_{\mathrm{cr}}$, with 
\begin{equation} \label{eq:DefTau}
    \tau_\star   \triangleq \beta_\star^{-1} \sup_{P \in \mathbb{S}^n_+} \sup_{i \geqslant n_{\mathrm{cr}} } \big\| \big[ I+ \mathcal{O}_{A_\star ,B_\star  }^{(i)}(P_\star) (P- P_\star) \big]^{-1} \big\|  .
\end{equation}

When $ 1 \leqslant i < n_{\mathrm{cr}}$, the rank of the controllability matrix is less than $n$, and thus $\mathcal{O}_{A_\star ,B_\star  }^{(i)}(P_\star)$ is not guaranteed to be of full rank. Then directly exploiting \eqref{eq:RiccatiDiffSameModel} leads to
\begin{align}  \| \mathcal{R}^{(i)}_{A_\star,B_\star}(P) - P_\star \| \leqslant
\|\Phi^{(0:i)}_{A_\star,B_\star}(P)\| \|L_\star ^i\| \|P-P_\star\|. \label{eq:KnownModelLemma2}
\end{align}  
A naive approach is to upper bound $\|\Phi^{(0:i)}_{A_\star,B_\star}(P)\|$ by a constant. However, due to $Q \succeq I$, $\|\Phi^{(0:i)}_{A_\star,B_\star}(P)\|$ decays exponentially as $i$ increases. This result in \eqref{eq:TransitionBoundSta} is exploited here and proved in Appendix~\ref{appdix:sta}. Combining \eqref{eq:RateLstar}, \eqref{eq:KnownModelLemma2}, and \eqref{eq:TransitionBoundSta} proves this case, with
\begin{equation} \label{eq:DefTauBar}
 \overline{\tau}_\star \triangleq    \Upsilon_\star^3 (1+ \Upsilon_\star+\|P_\star\|)  \sqrt{ \beta_\star^{-1} (\|P_\star\|+ \Upsilon_\star \beta_\star^{-1} ) (1 - \beta_\star)^{-1}  } ,
\end{equation}
derived from \eqref{eq:TransitionBoundSta} using $\Upsilon_\star \geqslant \varepsilon_{\mathrm{p}}$. The case $i=0$ holds trivially due to $ \overline{\tau}_\star \geqslant 1$.

\subsection{Proof of Lemma~\ref{lem:metalemma}}
The result is based on the following lemma:
\begin{lem} \label{lem:Kgap}
For any $F \in \mathbb{S}^n_+$ with $\|F-P_\star\| \leqslant \varepsilon$, if Assumption~\ref{ass:LQR} and \ref{ass:LQRgeneral2} hold, and if $\Upsilon_\star \geqslant \varepsilon$, then $K = \mathcal{K}_{\hat{A},\hat{B}}(F) $ satisfies 
\begin{equation} \label{eq:Kdiff}
\| K - K_\star  \| \leqslant \frac{   \Upsilon_\star^2 (\sqrt{\|P_\star\|}+1) (3 \varepsilon_\mathrm{m} +4 \varepsilon) }{ \underline{\sigma}(R)}.
\end{equation}
\end{lem}
\begin{proof}
Standard algebraic reformulations show $ \|\hat{B}^\top F \hat{B} - B_\star^\top P_\star B_\star \| \leqslant \varepsilon_\mathrm{m}^2 \varepsilon + 2\varepsilon_\mathrm{m} \varepsilon \Upsilon_\star + \|P_\star \| \varepsilon_\mathrm{m}^2+ 2 \varepsilon_\mathrm{m} \Upsilon_\star \|P_\star\|  +  \varepsilon \Upsilon_\star^2  \leqslant 4 \Upsilon_\star^2 \varepsilon + 3\Upsilon_\star \|P_\star\| \varepsilon_\mathrm{m},$ where the last inequality uses $\Upsilon_\star \geqslant \max\{ \varepsilon_\mathrm{m},\varepsilon\}$. Then following \cite[Lem. 2]{mania2019certainty} analogously, we have for any $x$, $ \underline{\sigma}(R) \| (K-K_\star)x\|  \leqslant (4 \Upsilon_\star^2 \varepsilon + 3\Upsilon_\star \|P_\star\| \varepsilon_\mathrm{m}) (\|K_\star\|+1)\|x\|  \leqslant     \Upsilon_\star^2 ( 4\varepsilon + 3 \varepsilon_\mathrm{m}) ( \sqrt{\|P_\star\|}+1)\|x\|,$
where the last inequality follows from Lemma~\ref{lem:PerturbationSimchow}.
\end{proof}

Then Lemma~\ref{lem:metalemma} can be proved as follows. Firstly, based on \eqref{eq:Kdiff} and the assumption, we have $\|B_\star (K-K_\star)\| \leqslant 8 \Upsilon_\star^4 ( \varepsilon_\mathrm{m} +\varepsilon)/  \underline{\sigma}(R) \leqslant 1/5 \|P_\star\|^{-3/2}$. Then given Assumption~\ref{ass:LQR}, we have $\|\Sigma_{K_\star}\| \leqslant \sigma_w^2 \|P_\star\|$ based on \cite[Lem. B.5]{simchowitz2020naive}\footnote{Note that in this work, the noise $w_t$ has a covariance matrix $\sigma_w^2 I$ instead of $I$ as in \cite{simchowitz2020naive}; however, the results in \cite{simchowitz2020naive} extends to this work trivially.}. Therefore, $\|B_\star (K-K_\star)\| \leqslant 1/5 \|P_\star\|^{-3/2}$ implies $\|B_\star (K-K_\star)\| \leqslant 1/5 \sigma_w^{3}  \|\Sigma_{K_\star}\|^{-3/2}$, which further shows that $A_\star-B_\star K$ is Schur stable, according to \cite[Lem. B.12]{simchowitz2020naive}. Then given a stabilizing $K$ and based on \cite[Lem.~3]{mania2019certainty}, we have 
\begin{align*}
   & J_{K} - J_{K_\star} \leqslant \|\Sigma_K\| \|R+ B_\star^\top P_\star B_\star \| \|K-K_\star\|_F^2 \\
   & \leqslant 2 \|\Sigma_{K_\star}\| \|R+ B_\star^\top P_\star B_\star \| \|K-K_\star\|^2 \min\{n,m\} \\
   & \leqslant 2 \min\{n,m\}  \sigma_w^2 (\overline{\sigma}(R) + \Upsilon_\star^3  ) \|P_\star\| \|K-K_\star\|^2, 
\end{align*}
where the second inequality holds due to $ \|\Sigma_K\|  \leqslant 2 \|\Sigma_{K_\star}\| $ under $\|B_\star (K-K_\star)\| \leqslant 1/5 \sigma_w^3 \|\Sigma_{K_\star}\|^{-3/2}$ according to \cite[Lem.~B.12]{simchowitz2020naive}. Then combining the above result and Lemma~\ref{lem:Kgap} concludes the proof.

\subsection{Proof of Lemma~\ref{lem:Pdiff}}
Define the shorthand notations $\hat{S} \triangleq S_{\hat{B}}$ and $S_\star \triangleq S_{B_\star}$. We prove this result by induction. The equality holds trivially when $i=0$. Assume it holds for $i=k$, then 
\begin{IEEEeqnarray}{rCl}
\label{eq:proofLemmaRiccati1a}
&&\hat{P}^{(k+1)}  - P_\star^{(k+1)} = \mathcal{R}^{(k)}_{\hat{A},\hat{B}}\big(\hat{P}^{(1)} \big) - \mathcal{R}^{(k)}_{A_\star,B_\star}\big( P_\star^{(1)} \big)  \nonumber    \\
& =& \big[ \Phi_{\hat{A},\hat{B}}^{(0:k)}(\hat{P}^{(1)} ) \big]^\top (\hat{P}^{(1)}  - P_\star^{(1)}) \bar{\Phi}^{(0:k)}(P_\star^{(1)})  \nonumber \\ && \negmedspace{}  +\sum_{j=1}^k  \big[ \Phi_{\hat{A},\hat{B}}^{(k-j+1:k)}(\hat{P}^{(1)}) \big]^\top  \mathcal{M}\Big( \hat{P}^{(k-j+1)} , P_\star^{(k-j+1)} \Big)  \nonumber \\
 && \negmedspace{} \times \bar{\Phi}^{(k-j+1:k)}(P_\star^{(1)}) \nonumber  \\
&=& \big[ \Phi_{\hat{A},\hat{B}}^{(1:k+1)}(P_1 ) \big]^\top (\hat{P}^{(1)}  - P_\star^{(1)}) \bar{\Phi}^{(1:k+1)}(P_2)  \nonumber \\  
&& \negmedspace{} 
 +\sum_{j=1}^k  \big[ \Phi_{\hat{A},\hat{B}}^{(k-j+2:k+1)}(P_1) \big]^\top  \mathcal{M}\Big( \hat{P}^{(k-j+1)} , P_\star^{(k-j+1)} \Big) \nonumber  \\
  && \negmedspace{}\times \bar{\Phi}^{(k-j+2:k+1)}(P_2) 
\end{IEEEeqnarray}

In \eqref{eq:proofLemmaRiccati1a}, it holds that $\hat{P}^{(1)}  - P_\star^{(1)} $
\begin{IEEEeqnarray}{rCl} 
&  =& \hat{A}^\top P_1 (I+\hat{S} P_1)^{-1}\hat{A} - A_\star^\top P_2 (I+S_\star P_2)^{-1}A_\star \nonumber \\
 &=& \underbrace{\hat{A}^\top [P_1 (I+\hat{S} P_1)^{-1}- P_2(I+S_\star P_2)^{-1}]\hat{A} }_{\text{Term I}} \nonumber \\
 &&\negmedspace{} + 
 \underbrace{\hat{A}^\top P_2 (I+S_\star P_2)^{-1}\hat{A}- A_\star^\top P_2 (I+S_\star P_2)^{-1}A_\star}_{\text{Term II}}, \IEEEeqnarraynumspace \label{eq:proofLemmaRiccati2} 
\end{IEEEeqnarray}
where Term I further leads to 
$
 \hat{A}^\top [P_1 (I+\hat{S} P_1)^{-1}- P_2 (I+S_\star P_2)^{-1}]\hat{A} =\hat{A}^\top [(I+P_1 \hat{S} )^{-1} P_1- P_2(I+S_\star P_2)^{-1}]\hat{A}  
 =  \hat{A}^\top (I+P_1 \hat{S} )^{-1}  [ P_1 (I+S_\star P_2)- (I+P_1 \hat{S} )P_2]    (I+S_\star P_2)^{-1}\hat{A} 
 =  \hat{A}^\top (I+P_1 \hat{S} )^{-1} [ (P_1-P_2) +(P_1 S_\star P_2 - P_1 \hat{S} P_2 )] (I+S_\star P_2)^{-1}\hat{A} 
 = \hat{A}^\top (I+P_1 \hat{S} )^{-1} (P_1-P_2) [I+(S_\star-\hat{S})P_2] (I+S_\star P_2)^{-1}\hat{A} 
 + \hat{A}^\top (I+P_1 \hat{S} )^{-1} P_2 (S_\star-\hat{S})P_2 (I+S_\star P_2)^{-1}\hat{A}$ 
\begin{align} \label{eq:proofLemmaRiccati2a}
 =& [\mathcal{L}_{\hat{A},\hat{B}}(P_1)]^\top (P_1-P_2) \bar{\mathcal{L}}(P_2) \nonumber\\ &+ \hat{A}^\top (I+P_1 \hat{S} )^{-1} P_2 (S_\star-\hat{S})P_2 (I+S_\star P_2)^{-1}\hat{A}. 
\end{align}
Combining \eqref{eq:proofLemmaRiccati2} and \eqref{eq:proofLemmaRiccati2a}  leads to 
\begin{equation}
\kern-0.6em \hat{P}^{(1)}  - P_\star^{(1)} =[\mathcal{L}_{\hat{A},\hat{B}}(P_1)]^\top (P_1-P_2) \bar{\mathcal{L}}(P_2)+ \mathcal{M}(P_1,P_2). \label{eq:proofLemmaRiccati3} 
\end{equation}
Then plugging \eqref{eq:proofLemmaRiccati3} into \eqref{eq:proofLemmaRiccati1a} shows $\hat{P}^{(k+1)}  - P_\star^{(k+1)}$ satisfies the equality in this lemma, concluding the proof by induction.

\subsection{Proof of Lemma~\ref{lem:PhiBar}}
The matrix $\bar{\mathcal{L}}( P_\star)$ in \eqref{eq:PhiBar} is reformulated as 
\begin{equation}
 \bar{\mathcal{L}}(P_\star)=  L_\star+ \underbrace{\mathcal{H}(P_\star) + S_\Delta  P_\star   [ \mathcal{H}(P_\star) + L_\star] }_{\Delta \triangleq}, \label{eq:PhiBarProof1}
\end{equation}
where $S_\Delta = S_{B_\star} - S_{\hat{B}}$. To upper bound the norm of $\Delta$, we first upper bound $\mathcal{H}(P_\star)$:
$ \| \mathcal{H}(P_\star)\| =\| [I - B_\star (R+ B_\star^\top P_\star B_\star)^{-1} B_\star^\top P_\star ] (\hat{A} - A_\star ) \| \leqslant (1 + \|B_\star\|^2 \|P_\star \| \| (R+ B_\star^\top P_\star B_\star)^{-1} \| ) \varepsilon_\mathrm{m} \leqslant (1 + \Upsilon_\star^2 \|P_\star \|  )\varepsilon_\mathrm{m},$
where the first identity follows form the matrix inverse lemma, and the final step follows from $R+ B_\star^\top P_\star B_\star \succeq R \succeq I$ due to Assumption~\ref{ass:LQR}. Due to $R \succeq I$, letting $\varepsilon_s = \|S_\Delta\| $ leads to 
\begin{equation} \label{eq:Sdelta}
    \varepsilon_s = \| S_{B_\star} - S_{\hat{B}}\|  \leqslant \varepsilon_\mathrm{m}^2 +2 \Upsilon_\star \varepsilon_\mathrm{m}, \text{ and }
\end{equation}
\begin{IEEEeqnarray}{rCl}
 \|\Delta\| &\leqslant& (1 + \Upsilon_\star^2 \|P_\star \|  )\varepsilon_\mathrm{m}+\varepsilon_s \|P_\star\| (\|\mathcal{H}(P_\star)
\|+ \|L_\star\| ) \nonumber \\
&\leqslant&  (1 + \Upsilon_\star^2 \|P_\star \|  )\varepsilon_\mathrm{m}  +\varepsilon_s  \|P_\star\| \big[ (1 + \Upsilon_\star^2 \|P_\star \|  )\varepsilon_\mathrm{m}+ \|P_\star\|  \big] \nonumber  \\ 
 &\leqslant&  \psi(\varepsilon_\mathrm{m}) \triangleq \varepsilon_{\mathrm{m}}  \Big\{ 1 + \Upsilon_\star^2 \|P_\star \| +(\varepsilon_\mathrm{m}+ 2 \Upsilon_\star )\nonumber  \\
 && \negmedspace{} \quad \quad\quad \quad \times  \Big[ (\|P_\star\| + \Upsilon_\star^2 \|P_\star \|^2  )\varepsilon_\mathrm{m}+ \|P_\star\|^2  \Big]\Big\}, \label{eq:psiDEFproof}
\end{IEEEeqnarray}
where we have used $\|L_\star \| \leqslant   \|P_\star \|$ from Lemma~\ref{lem:PerturbationSimchow} in the second inequality, and $\psi(\varepsilon_\mathrm{m}) =O(\varepsilon_{\mathrm{m}}) $ holds. Based on \eqref{eq:RateLstar} and $\bar{\Phi}^{(j:i)}(P_\star)  = \bar{\mathcal{L}}(P_\star)^{i-j} = (L_\star + \Delta)^{i-j}$, applying Lemma~\ref{lem:PerturbRate} concludes the proof.

\subsection{Proof of Theorem~\ref{thm:RiccatiDiff}}
Based on \eqref{eq:lemmaRiccati}, we first upper bound $\mathcal{M}$ and then combine it with Lemma~\ref{lem:PhiBar} and Proposition~\ref{prop:PhiBarAB}. Recall \eqref{eq:Sdelta}, and for any $P \in \mathbb{S}^n_+$, $\|\mathcal{M}(P,P_\star)\| \leqslant \| (\hat{A}-A_\star)^\top P_\star (I+S_{B_\star}P_\star)^{-1} (\hat{A}-A_\star)+ A_\star^\top P_\star (I+S_{B_\star}P_\star)^{-1} (\hat{A}-A_\star)+  (\hat{A}-A_\star)^\top P_\star (I+S_{B_\star}P_\star)^{-1}  A_\star \|+ \|\hat{A}\|  \|P_\star\|^2 \|\mathcal{L}_{\hat{A}, \hat{B}}(P)     \|  \varepsilon_s$
\begin{IEEEeqnarray}{rCl} 
    &\leqslant & \|P_\star\| ( \varepsilon_\mathrm{m}^2 +2\Upsilon_\star \varepsilon_\mathrm{m}  ) +(\varepsilon_\mathrm{m}+ 
\Upsilon_\star)^4 \|P_\star\|^2   (1+ \|P\|)    \varepsilon_s \nonumber  \\
 &\leqslant & \Tilde{\psi}(\varepsilon_\mathrm{m} ) \triangleq ( \varepsilon_\mathrm{m}^2 +2\Upsilon_\star \varepsilon_\mathrm{m}  ) \nonumber \\
&& \negmedspace{}  \times \big[ \|P_\star\|+ (\varepsilon_\mathrm{m}+ 
 \Upsilon_\star)^4 \|P_\star\|^2   (1+ \|P_\star\|+\Upsilon_\star )   \big],  \label{eq:TidlePsiIneq}
\end{IEEEeqnarray}
where \eqref{eq:BoundCloseLoop} and Lemma~\ref{lem:tool1Mania} are used,
 and $\Tilde{\psi}(\varepsilon_\mathrm{m} ) = O(\varepsilon_\mathrm{m})$.

Then if $i \geqslant n_{\mathrm{cr}}$, combining \eqref{eq:lemmaRiccati}, \eqref{eq:PhiBarDecayRate}, \eqref{eq:PhihatIneq} and \eqref{eq:TidlePsiIneq} shows 
\begin{equation*}
\| \hat{P}^{(i)}  - P_\star  \|  \leqslant \Tilde{\zeta}\varepsilon_{\mathrm{p}} \gamma_2^i \gamma_1^i  + \Tilde{\zeta}   \Tilde{\psi} \Big (   \sum_{k=0}^{n_{\mathrm{cr}}-1} \overline{\gamma}_2^k \gamma_1^k+  \sum_{k=n_{\mathrm{cr}}}^{i-1} \gamma_2^k \gamma_1^k  \Big).\end{equation*}
Similarly, when $i < n_{\mathrm{cr}}$, we have
$
\| \hat{P}^{(i)}  - P_\star  \|  \leqslant \Tilde{\zeta}\varepsilon_{\mathrm{p}} \overline{\gamma}_2^i \gamma_1^i + \Tilde{\zeta}   \Tilde{\psi}   \sum_{k=0}^{i-1} \overline{\gamma}_2^k \gamma_1^k.
$ Combining these two cases proves the result.

\subsection{ Proof of Lemma~\ref{prop:MPCperformance} }
We have $\widehat{E} \geqslant \Tilde{\zeta} \Tilde{\psi} \geqslant 2 \varepsilon_{\mathrm{m}} \Upsilon_\star^4 \|P_\star\|^2$, due to 
\begin{equation}
\Tilde{\zeta} \geqslant \overline{\zeta} \geqslant \Upsilon_\star^3 \|P_\star\|, \text{ and }\Tilde{\psi} \geqslant 2 \varepsilon_{\mathrm{m}} \Upsilon_\star \|P_\star\|,
\end{equation}
which are obtained from \eqref{eq:PhiConstant2} and \eqref{eq:TidlePsiIneq}. The assumption $\varepsilon_\mathrm{m}+ \widehat{E}(\varepsilon_\mathrm{m}, \varepsilon_\mathrm{p},N-1) \leqslant 1/(40 \Upsilon_\star^4 \|P_\star\|^2 )$ implies $\varepsilon_{\mathrm{m}} \leqslant 1/(80 \Upsilon_\star^8 \|P_\star\|^4 )$. Then \eqref{eq:Defgamma1Psi} and the definition of $\psi$ in \eqref{eq:psiDEFproof} show $\psi \leqslant \varepsilon_\mathrm{m} 9 \Upsilon_\star^3 \|P_\star\|^2 \leqslant 1/(8 \Upsilon_\star^5 \|P_\star\|^2 )$
\begin{align}
    & <\frac{1}{2\|P_\star\|^{3/2}}\leqslant \frac{1}{(1+\sqrt{1-\beta_\star})\sqrt{\beta_\star^{-1}} \beta_\star^{-1} }= \frac{1- \sqrt{1-\beta_\star }}{ \sqrt{\beta_\star^{-1}}} \label{eq:boundpsi }\\
    & \implies \gamma_1 < 1, \nonumber
\end{align}
where we use $\sqrt{1-\beta_\star}<1$ and $\beta_\star^{-1}  \leqslant \|P_\star\|$. combining Lemma~\ref{lem:metalemma} and Theorem~\ref{thm:RiccatiDiff} concludes the proof, as the conditions in these two results are satisfied by 
$\varepsilon_{\mathrm{m}} \leqslant 1/(80 \Upsilon_\star^8 \|P_\star\|^4 )$.

\subsection{Proof of Theorem~\ref{coro:ChangeN}}
Equation~\eqref{eq:PerturbRiccatiDiff} can be written more compactly as
\begin{align}
\widehat{E} =& \Tilde{\zeta} \Big \{ \Big( \varepsilon_\mathrm{p} - \frac{\Tilde{\psi} }{ 1- \gamma_1 \gamma_2   } \Big)\gamma_2^{N-1} \gamma_1^{N-1} + \frac{\Tilde{\psi} \gamma_1^{n_{\mathrm{cr}}} \gamma_2^{n_{\mathrm{cr}}} }{1-  \gamma_1 \gamma_2  } \nonumber \\& +\Tilde{\psi} \sum_{k=0}^{n_{\mathrm{cr}}-1} \overline{\gamma}_2^k\gamma_1^k  \Big\} \text{ if } N-1 \geqslant n_{\mathrm{cr}}.\label{eq:EreformuB} 
\end{align}
We can quantify the change from $N-1<n_{\mathrm{cr}}$ to  $N-1 \geqslant n_{\mathrm{cr}}$ by computing $\Delta\widehat{E} \triangleq \widehat{E}(\varepsilon_\mathrm{m}, \varepsilon_\mathrm{p},n_{\mathrm{cr}} ) -  \widehat{E}(\varepsilon_\mathrm{m}, \varepsilon_\mathrm{p},n_{\mathrm{cr}}-1 ) $
\begin{align*}
& = \Tilde{\zeta}   \varepsilon_\mathrm{p}\big( \gamma_1^{n_{\mathrm{cr}}} \gamma_2^{ n_{\mathrm{cr}}} - \gamma_1^{n_{\mathrm{cr}}-1} \overline{\gamma}_2^{ n_{\mathrm{cr}}-1} \big) +\Tilde{\psi}  \Tilde{\zeta} \gamma_1^{n_{\mathrm{cr}}-1} \overline{\gamma}_2^{ n_{\mathrm{cr}}-1}   \\
 & = \Tilde{\zeta} (\gamma_1 \gamma_2)^{n_{\mathrm{cr}}-1}\big\{\varepsilon_{\mathrm{p}}[ \gamma_2 \gamma_1 - (\overline{\gamma}_2/\gamma_2)^{ n_{\mathrm{cr}}-1 }]   +  \Tilde{\psi} (\overline{\gamma}_2/\gamma_2)^{ n_{\mathrm{cr}}-1 } \big\}.
\end{align*}
The above equation and $\gamma_2 \leqslant \overline{\gamma}_2$ from \eqref{eq:PhihatIneq} show that 
\begin{equation}\label{eq:proofRateCompareC}
\Delta\widehat{E} \geqslant 0 \text{ iff }   \varepsilon_\mathrm{p} - \frac{\Tilde{\psi} }{ 1-\frac{ \gamma_1  \gamma_2}{ (\overline{\gamma}_2/\gamma_2)^{n_{\mathrm{cr}}-1 }     } } \leqslant 0.
\end{equation}
 
The statement (a) follows from \eqref{eq:Ereformu}, \eqref{eq:proofRateCompare}, \eqref{eq:EreformuB}, and \eqref{eq:proofRateCompareC}. For statement (b), if $ 0 >\varepsilon_\mathrm{p} - \Tilde{\psi}/(1-\gamma_1 \gamma_2 )$, then \eqref{eq:Ereformu}, \eqref{eq:proofRateCompare}, and \eqref{eq:EreformuB} show that $g$ and $\widehat{E}$ are strictly increasing as $N \in \{1,\dots n_{\mathrm{cr}}\}$ increases, and they are also strictly increasing as $N \geqslant n_{\mathrm{cr}}+1$ increases. For the transition of the two regions of $N$, consider $\widehat{E} (\varepsilon_\mathrm{m},\varepsilon_\mathrm{p},0) -  \widehat{E} (\varepsilon_\mathrm{m},\varepsilon_\mathrm{p},n_{\mathrm{cr}})$
\begin{IEEEeqnarray*}{rCl}
& =& \varepsilon_\mathrm{p} \Tilde{\zeta} -  \varepsilon_\mathrm{p} \Tilde{\zeta} (\gamma_1 \gamma_2)^{n_{\mathrm{cr}} } -  \Tilde{\zeta} \Tilde{\psi} \frac{1- (\gamma_1 \overline{\gamma}_2)^{n_{\mathrm{cr}}}  }{1- \gamma_1 \overline{\gamma}_2 } =\Tilde{\zeta} [1- (\gamma_1 \gamma_2)^{n_{\mathrm{cr}} } ]\\
&&\negmedspace{} \times  \Big\{\varepsilon_\mathrm{p} -\frac{\Tilde{\psi}}{ 1- \gamma_1\gamma_2  }   \frac{ (1-\gamma_1 \gamma_2) [1- (\gamma_1 \overline{\gamma}_2)^{n_{\mathrm{cr}}}]  }{ (1- \gamma_1 \overline{\gamma}_2) [1- (\gamma_1 \gamma_2)^{n_{\mathrm{cr}} } ]   }  \Big \} \\
&\leqslant& \Tilde{\zeta} [1- (\gamma_1 \gamma_2)^{n_{\mathrm{cr}} } ]  
 \Big( \varepsilon_\mathrm{p} -\frac{\Tilde{\psi}}{ 1- \gamma_1\gamma_2  }    \Big ) <0,
\end{IEEEeqnarray*} 
where the first inequality follows from Lemma~\ref{lem:inequalProof}. The above shows $\arg\min_N \widehat{E}(\varepsilon_\mathrm{m},\varepsilon_\mathrm{p},N-1) =\arg\min_N g(\varepsilon_\mathrm{m},\varepsilon_\mathrm{p},N) = 1$. Moreover, under the additional condition $ \varepsilon_\mathrm{p} - \Tilde{\psi}/[1-\gamma_1 \gamma_2/\gamma_{\mathrm{ratio}} ] < 0$, the last part of statement (b) is proved by observing \eqref{eq:proofRateCompareC}.

For statement (c), if $\varepsilon_\mathrm{p} - \Tilde{\psi}/(1-\gamma_1 \gamma_2 ) >0> \varepsilon_\mathrm{p} - \Tilde{\psi}/(1-\gamma_1  \overline{\gamma}_2)$, \eqref{eq:Ereformu} and \eqref{eq:EreformuB} show that $g$ is increasing as $N \in \{1,\dots, n_{\mathrm{cr}}\}$ increases and decreasing as $N \geqslant n_{\mathrm{cr}}+1$ increases. This means that $\inf_{N} g(\varepsilon_\mathrm{m},\varepsilon_\mathrm{p},N) $ equals either $g(\varepsilon_\mathrm{m},\varepsilon_\mathrm{p},1) $ or $\lim_{N \to \infty}g(\varepsilon_\mathrm{m},\varepsilon_\mathrm{p},N) $ depending on the sign of $\widehat{E} (\varepsilon_\mathrm{m},\varepsilon_\mathrm{p},0) -  \lim_{N\to \infty }\widehat{E} (\varepsilon_\mathrm{m},\varepsilon_\mathrm{p},N)=$
\begin{align*}
& \Tilde{\zeta}\bigg\{
\varepsilon_\mathrm{p} - \frac{\Tilde{\psi}}{ 1- \gamma_1\gamma_2  } \Big[ \underbrace{\frac{(1- \gamma_1\gamma_2)[1- (\overline{\gamma}_2 \gamma_1)^{n_{\mathrm{cr}}}]}{1- \gamma_1\overline{\gamma}_2} +(\gamma_1 \gamma_2)^{n_{\mathrm{cr}}}}_{\mathcal{V}(n_{\mathrm{cr}}) 
 } 
 \Big]
\bigg\},
\end{align*}
which can be proved to be in the interval $ [ \Tilde{\zeta}(\varepsilon_\mathrm{p} - \Tilde{\psi}/(1-\gamma_1  \overline{\gamma}_2)), \Tilde{\zeta}(\varepsilon_\mathrm{p} - \Tilde{\psi}/(1-\gamma_1 \gamma_2 ))]$ as follows. The upper bound of this interval is proved by showing $\mathcal{V}(i+1) - \mathcal{V}(i) \geqslant 0$ and thus $\mathcal{V}(n_{\mathrm{cr}}) \geqslant \mathcal{V}(1)=1$, and the lower bound can be proved similarly. The above proves statement (c). Statement (d) follows trivially from \eqref{eq:Ereformu} and \eqref{eq:EreformuB}.

\subsection{Proof of Corollary~\ref{prop:MPCnoTerminal}}
When $P=0$, we let $\varepsilon_\mathrm{p} = \|P_\star\|$. According to Theorem~\ref{coro:ChangeN}, the goal is to prove $\|P_\star\|- \Tilde{\psi}/(1-\gamma_1 \overline{\gamma}_2) >0 $ under the conditions. Due to $\Tilde{\zeta} \geqslant \overline{\zeta} \geqslant 1$, $\Tilde{\psi} \leqslant \widehat{E} \leqslant 1/(40 \Upsilon_\star^4 \|P_\star\|^2) $ holds, given the condition and \eqref{eq:PerturbRiccatiDiff}. Based on $\psi \leqslant 1/(8 \Upsilon_\star^5 \|P_\star\|^2 )$ from \eqref{eq:boundpsi }, we have from \eqref{eq:Defgamma1Psi} that
$
  \gamma_1  \leqslant \sqrt{\|P_\star\|} \psi + \sqrt{1-\beta_\star} < 1/ 8 \Upsilon_\star^5 \|P_\star\|^{3/2}.
$
This leads to
\begin{align*}
 1-\gamma_1 -\Tilde{\psi} > 1- \frac{1}{8 \Upsilon_\star^5 \|P_\star\|^{3/2} }  -\frac{1}{40 \Upsilon_\star^4 \|P_\star\|^2}>0,  
\end{align*}
and $\Tilde{\psi} /(1-\gamma_1)<1$. Finally, $\|P_\star\| \geqslant 1 $ from Lemma~\ref{lem:PerturbationSimchow} implies
\begin{align}
 &\|P_\star\|- \frac{\Tilde{\psi}}{ 1- \gamma_1 \overline{\gamma}_2  } \geqslant   1 - \frac{\Tilde{\psi}}{1-\gamma_1} >0, \label{eq:ZeroTermi}
\end{align}
which together with Theorem~\ref{coro:ChangeN} concludes the proof.

\subsection{Proof of Proposition~\ref{lem:SimplePerfor}}
Since $P=0$, we let $\varepsilon_\mathrm{p}= \|P_\star\|$. In addition, $\varepsilon_\mathrm{m }\leqslant 1/(16 \|P_\star\|^2)$ holds under the assumption, which leads to 
\begin{equation} \label{eq:BoundALpha}
   \alpha_{\varepsilon_\mathrm{m}} \in [1, \sqrt{2}], 
\end{equation}
and thus $ \overline{\gamma}_2 \leqslant \mu_\star<1$. Therefore, we have 
\begin{align}
& \widehat{E}_{\mathrm{sta}} \leqslant \Tilde{\zeta} \Big[ \|P_\star\|  \mu_\star ^{N- 1}+ \Tilde{\psi}/(1- \mu_\star)     \Big].\label{eq:proofSimplePerfor}
\end{align}
Given the upper bound on $\varepsilon_{\mathrm{m}}$, $\Tilde{\psi}$ in \eqref{eq:TidlePsiIneq} satisfies $\Tilde{\psi} \leqslant C_1 \varepsilon_\mathrm{m} $ for some constant $C_1$ independent of $N$ and $\varepsilon_{\mathrm{m}}$, and $\Tilde{\zeta}$ can also be upper bounded by some constant $C_2$. This shows $
\widehat{E}_{\mathrm{sta}}  \leqslant C_2  \max\{\|P_\star\| ,C_1/(1-\gamma_1 \mu_\star)\}   [ \mu_\star ^{N- 1} +\varepsilon_\mathrm{m} ].
$
Combining the above inequality with \eqref{eq:MPCperforUnknown} concludes the proof.

 \section{Proof of Proposition~\ref{prop:PhiBarAB}\label{sec:Appendix2} }
To prove this result, the dual Riccati mapping is relevant:
\begin{align}
  \mathcal{R}_{ \mathrm{dual}}(P)  \triangleq & A_\star P (I+ QP)^{-1}A_\star^\top +  S_{B_\star},   \label{eq:DualRicca} 
\end{align}
with its fixed point $P_{\mathrm{dual}}$. Let $\hat{P}_{\mathrm{dual}}$ be a fixed point of \eqref{eq:DualRicca} with $(A_\star,B_\star)$ replaced by $(\hat{A},\hat{B})$ and $\hat{L}$ be the closed-loop matrix for $(\hat{A},\hat{B})$ under the corresponding LQR controller.

The cases $i-j \geqslant n_{\mathrm{cr}} $ and $i-j< n_{\mathrm{cr}}$ are considered separately. If $i-j \geqslant n_{\mathrm{cr}} $, the proof strategy is as follows. Due to $\Phi_{\hat{A},\hat{B}}^{(j:i)}(P) = \Phi_{\hat{A},\hat{B}}^{(0:i-j)}(\hat{P}^{(j)})$, Lemma~\ref{lem:RicaDiffEqu}\ref{eq:IdentiPhi} shows $\big[ I+ \mathcal{O}_{\hat{A} ,\hat{B} }^{(i-j)}(\hat{P}  ) (\hat{P}^{(j)}- \hat{P}) \big] 
 \Phi^{(j:i)}_{\hat{A},\hat{B}}(P)=  \hat{L}^{i-j}$ for $i \in \mathbb{Z}^+$. Moreover, \cite[Lem.~4.1]{del2021note} shows that if  $ \mathcal{O}_{\hat{A} ,\hat{B} }^{(i-j)}(\hat{P} )$ is positive definite, $\|\Phi^{(j:i)}_{\hat{A},\hat{B}}(P)\| =   \| \big[ I+ \mathcal{O}_{\hat{A} ,\hat{B} }^{(i-j )}(\hat{P} ) (\hat{P}^{(j)}- \hat{P}) \big]^{-1} \hat{L}^{i-j} \| $
 \begin{align}
 \leqslant \| [\mathcal{O}_{\hat{A} ,\hat{B} }^{(i-j )}(\hat{P} ) ]^{-1} \|  \|\hat{P}_{\mathrm{dual}}\| \| \hat{L}^{i-j} \|, \label{eq:ProofStateTransi}
\end{align}
Then the proof consists of the following analysis: (i) Perturbation analysis of the closed-loop system $\hat{L}$; (ii) Perturbation analysis of the Gramian matrix: Show that $\mathcal{O}_{\hat{A} ,\hat{B} }^{(i-j)}(\hat{P} )$ is invertible for $i-j \geqslant n_{\mathrm{cr}} $ and provides an upper bound on $\| [\mathcal{O}_{\hat{A} ,\hat{B} }^{(i-j)}(\hat{P} ) ]^{-1} \|$; (iii) Perturbation analysis of $\|\hat{P}_{\mathrm{dual}}\|$. All perturbation bounds will be functions of $\varepsilon_{\mathrm{m}}$ and $\varepsilon_\mathrm{p}$. 

The decay rate of the perturbed closed-loop system is obtained here. Given $Q \succeq I$, \eqref{eq:RateLstar} and Lemma~\ref{lem:PerturbationSimchow} lead to
\begin{align}
  \|\hat{L} ^{i-j}\|  \leqslant& \sqrt{ \overline{\sigma}(\hat{P})/ \underline{\sigma}(Q)   }  \Big(\sqrt{ 1-  \underline{\sigma}(Q)/  \overline{\sigma}( \hat{P})  } \Big)^{i-j} \nonumber \\
  \leqslant & \sqrt{\alpha_{\varepsilon_\mathrm{m}}  \beta_\star^{-1} }  \Big(\sqrt{ 1- \alpha_{\varepsilon_\mathrm{m}}^{-1} \beta_\star    } \Big)^{i-j}. \label{eq:LhatRate} 
\end{align}
The remaining analysis is presented as follows.

\subsection{Perturbation analysis of the Gramian matrix}
If $i-j \geqslant n_{\mathrm{cr}}$, then $\mathcal{O}_{\hat{A} ,\hat{B} }^{(i-j)}(\hat{P} ) \succeq 
 \mathcal{O}_{\hat{A} ,\hat{B} }^{(n_{\mathrm{cr}})}(\hat{P} )$ and 
 $$ \mathcal{O}_{\hat{A} ,\hat{B} }^{(n_{\mathrm{cr}})}(\hat{P} ) \succeq   (\|R+\hat{B}^\top \hat{P} \hat{B} \|)^{-1} \mathcal{C}_{n_{\mathrm{cr}}}(\hat{L},\hat{B}) \big[\mathcal{C}_{n_{\mathrm{cr}}}(\hat{L},\hat{B}) \big]^\top.$$ The Gramian matrix $  \mathcal{C}_{n_{\mathrm{cr}}}(\hat{L},\hat{B})[\mathcal{C}_{n_{\mathrm{cr}}}(\hat{L},\hat{B})]^\top $ is a perturbed version of $  \mathcal{C}_{n_{\mathrm{cr}}}( L_\star, B_\star)[\mathcal{C}_{n_{\mathrm{cr}}}(L_\star,B_\star)]^\top $. Based on Lemma~\ref{lem:PerturbationSimchow}, $\|\hat{L} - L_\star\| \leqslant \|A_\star-\hat{A}\| + \|(\hat{B}- B_\star) ( \hat{K} - K_\star) \|  + \|(\hat{B} - B_\star) K_\star\|   + \|B_\star(\hat{K} - K_\star) \|$
 \begin{align}
  & \leqslant 
 \varepsilon_L   \triangleq \varepsilon_{\mathrm{m}}\big[1+ 7 \alpha_{\varepsilon_\mathrm{m}}^{7/2} \|P_\star\|^{7/2}(\varepsilon_{\mathrm{m}}+\Upsilon_\star) + \|P_\star\|^{1/2}  \big],  \label{eq:Lperburb}
 \end{align}
which satisfies
\begin{equation} \label{eq:rateEpsilonL}
 \varepsilon_L  = O(\varepsilon_\mathrm{m} ).   
\end{equation}

Under Assumption~\ref{ass:ControlIndex}, recall $
 \underline{\sigma}(\mathcal{C}_{n_{\mathrm{cr}}}(L_\star,B_\star)) \geqslant \nu_\star $ from \eqref{eq:defGramianLowerBound}. This together with \eqref{eq:RateLstar}, \eqref{eq:Lperburb}, and \cite[Lem. 6]{mania2019certainty} show  $
\underline{\sigma}( \mathcal{C}_{n_{\mathrm{cr}}}(\hat{L},\hat{B}) ) \geqslant  \nu_\star - f_{\mathcal{C}}(\varepsilon_{\mathrm{m}} ),$
with $f_{\mathcal{C}}(\varepsilon_{\mathrm{m}} )   \triangleq    3 \varepsilon_{L} n_{\mathrm{cr}}^{3/2} \beta_\star^{-1} $
\begin{align} \label{eq:defFc}
& \times \max\big\{1,  \sqrt{\beta_\star^{-1}}  \varepsilon_{L} +\sqrt{ 1-\beta_\star }     \big\}^{n_{\mathrm{cr}}-1} (\Upsilon_\star + 1)   \\
&= O(\varepsilon_\mathrm{m}), \nonumber
\end{align}
where $\varepsilon_{L}$ is defined in \eqref{eq:Lperburb}, and the second equality follows from \eqref{eq:rateEpsilonL}. Given $\nu_\star \geqslant  2 f_{\mathcal{C}}(\varepsilon_{\mathrm{m}} ) $ in Assumption~\ref{ass:TechContPertur}, $\mathcal{C}_{n_{\mathrm{cr}}}(\hat{L},\hat{B}) $ is of full rank, and we have  $\lambda_{\mathrm{min}}( \mathcal{C}_i(\hat{L},\hat{B})\mathcal{C}_i(\hat{L},\hat{B})^\top ) \geqslant (\nu_\star - f_{\mathcal{C}}(\varepsilon_{\mathrm{m}} ))^2$. Algebraic manipulations show $\|\hat{B}^\top \hat{P} \hat{B}\| \leqslant \|\hat{P}\|(\varepsilon_{\mathrm{m}}+ \Upsilon_\star  )^2$, which leads to $[ \mathcal{O}_{\hat{A} ,\hat{B} }^{( n_{\mathrm{cr}} )}(\hat{P} )]^{-1} \preceq	    (\|R+\hat{B}^\top \hat{P} \hat{B} \|)   (\nu_\star -f_{\mathcal{C}}(\varepsilon_{\mathrm{m}} ))^{-2} I$
\begin{subequations}
\begin{align}   
&\preceq  \big[  \overline{\sigma}(R) + \alpha_{\varepsilon_\mathrm{m}} \| P_\star\|(\varepsilon_{\mathrm{m}}+ \Upsilon_\star  )^2   \big]  (\nu_\star - f_{\mathcal{C}}(\varepsilon_{\mathrm{m}} ))^{-2} I, \label{eq:PerturGramianLowerBound}\\
&  = O(1) I, \label{eq:PerturGramianLowerBoundb}
\end{align}
\end{subequations}
where \eqref{eq:PerturGramianLowerBoundb} follows from  the fact that the RHS of \eqref{eq:PerturGramianLowerBound} approaches a constant as $\varepsilon_\mathrm{m} \to 0$.
\subsection{Perturbation analysis of the dual Riccati equation} \label{sec:DualPerturb}
Given $Q \succeq  I$ under Assumption~\ref{ass:LQR}, let $Q^{1/2}$ denote the symmetric positive-definite square root of $Q$. We further define 
\begin{align}
 K_{\mathrm{dual}} &\triangleq (I+  Q^{\frac{1}{2}} P_{\mathrm{dual}}   Q^{\frac{1}{2}}  )^{-1} Q^{\frac{1}{2}} P_{\mathrm{dual}} A_\star^\top, \nonumber \\
  L_{\mathrm{dual}} &\triangleq A_\star^\top -  Q^{\frac{1}{2} }K_{\mathrm{dual}}. \label{eq:LdualDef}
\end{align}
As $  L_{\mathrm{dual}} $ is Schur stable, let $\rho_{\mathrm{dual}} \in (0,1)$ represent its decay rate such that $\|L_{\mathrm{dual}}^i\| \leqslant c \rho_{\mathrm{dual}}^i $ for some $c \geqslant 1$ and for any $i \in \mathbb{Z}^+$. Then applying \cite[Prop.~2]{mania2019certainty} leads to the following:
\begin{lem} \label{lem:DualRccatiPerturb}
It holds that
$$
\|\hat{P}_{\mathrm{dual}} - P_{\mathrm{dual}}\| \leqslant O(1)  \frac{\rho_{\mathrm{dual}}^2}{1-\rho_{\mathrm{dual}}^2 } \Upsilon_\star (\Upsilon_\star+1)^3 \|P_{\mathrm{dual}}   \|^2 \varepsilon_{\mathrm{m}},
$$
if Assumptions~\ref{ass:LQR}, \ref{ass:LQRgeneral2}, and $\underline{\sigma}(P_{\mathrm{dual}}) \geqslant 1$ hold, and if 
\begin{align}
\varepsilon_{\mathrm{m}} \leqslant \alpha_\star \triangleq   & \text{ } O(1)  \frac{(1-\rho_{\mathrm{dual}}^2)^2}{ \rho_{\mathrm{dual}}^4 (\Upsilon_\star+1)^{5}  \|P_{\mathrm{dual}}   \|^{2}} \nonumber  \\   
&\times \min\{(\|L_{\mathrm{dual}}\|+1)^{-2}, \|P_{\mathrm{dual}}   \|^{-1} \}. \label{eq:DualPerturbCondi}
\end{align}
\end{lem}
\begin{proof}
Based on Assumptions~\ref{ass:LQR} and \ref{ass:LQRgeneral2}, $(A_\star^\top, Q^{1/2})$ is stabilizable, and $(A_\star^\top, R^{-1/2}B_\star^\top )$ is observable, showing $\mathcal{R}_{ \mathrm{dual}}$ has a unique positive-definite fixed point $P_{\mathrm{dual}}$. Moreover, \eqref{eq:Sdelta} and the assumption $\Upsilon_\star \geqslant \varepsilon_{\mathrm{m}}$ imply $\max\{\|\hat{A}^\top-A_\star^\top\|, \|Q^{1/2} - Q^{1/2} 
 \|, \|  S_{\hat{B}} - S_{B_\star}\| \} \leqslant 3 \Upsilon_\star \varepsilon_{\mathrm{m}}$. Then applying\cite[Prop.~2]{mania2019certainty} proves the current lemma. 
\end{proof}

In Lemma~\ref{lem:DualRccatiPerturb}, $O(1)$ is a constant and is independent of $\varepsilon_{\mathrm{m}}$. Also note that the Riccati perturbation result from \cite{simchowitz2020naive} does not apply to the dual Riccati equation \eqref{eq:DualRicca} as the result requires the condition $S_{B_\star} = B_\star R^{-1} B_\star^\top \succeq I$, which is not met by many practical situations.
\begin{rem} \label{rem:Dual}
Assumption $\underline{\sigma}(P_{\mathrm{dual}}) \geqslant 1$ relates to controllability of $(L_{\mathrm{dual}}^\top,B_\star R^{1/2})$ with $L_{\mathrm{dual}}$ defined in \eqref{eq:LdualDef}. Comparing Lemma~\ref{lem:RicaDiffEqu}\ref{eq:IdentiRicatiLyapu} and \eqref{eq:DualRicca},
$
 P_{\mathrm{dual}} \succeq \sum_{i=0}^\infty (L_{\mathrm{dual}}^\top)^i S_{B_\star} L_{\mathrm{dual}}^i
 $ holds.
 Therefore, $\underline{\sigma}(P_{\mathrm{dual}}) \geqslant 1$ can be guaranteed by lower bounding the controllability Gramian of $(L_{\mathrm{dual}}^\top,B_\star R^{1/2})$.
\end{rem}

\subsection{Final step when $i-j \geqslant n_{\mathrm{cr}}$}
When $i-j \geqslant n_{\mathrm{cr}}$, the perturbation analysis can be concluded by combining \eqref{eq:ProofStateTransi}, \eqref{eq:LhatRate}, \eqref{eq:PerturGramianLowerBound}, and Lemma~\ref{lem:DualRccatiPerturb}: 
\begin{align}
    \|\Phi^{(j:i)}_{\hat{A},\hat{B}}(P)\| \leqslant  \zeta   \big( \sqrt{1- \alpha^{-1}_{\varepsilon_\mathrm{m}} \beta_\star } \big)^{i-j}, \text{ with }
\end{align}
\begin{align}
\hspace{-0.2cm} \zeta \triangleq  \big[  \overline{\sigma}(R) + \sqrt{2} \| P_\star\|(2 \Upsilon_\star )^2   \big] [ 1/2 \nu_\star  ]^{-2}  \sqrt{ \sqrt{2} \beta_\star^{-1} }  \nonumber\\
\hspace{-0.2cm}  \times \Big[ \|P_{\mathrm{dual}}\| +  O(1)  \frac{\rho_{\mathrm{dual}}^2}{1-\rho_{\mathrm{dual}}^2 } \Upsilon_\star (\Upsilon_\star+1)^3 \|P_{\mathrm{dual}}   \|^2 \Upsilon_\star \Big ],     \label{eq:PhiZeta}
\end{align}
where we have used \eqref{eq:BoundALpha} and $\Upsilon_\star \geqslant \max\{\varepsilon_{\mathrm{m}},\varepsilon_{\mathrm{p}} \}$.

\subsection{Situation where $i-j < n_{\mathrm{cr}}$  }
If $i-j < n_{\mathrm{cr}}$, $\mathcal{O}_{\hat{A} ,\hat{B} }^{(i-j)}(\hat{P} )$ may not be invertible, and thus instead of exploiting controllability and \eqref{eq:ProofStateTransi}, we consider the alternative bound in \eqref{eq:PerturbationStabilizable}, which holds for stabilizable systems.
 
 \section{Proof of Theorem~\ref{thm:stabi} } \label{appdix:sta}
Based on Lemma~\ref{lem:Pdiff}, an essential step to prove Theorem~\ref{thm:stabi} is establishing the convergence rate of the state-transition matrix for stabilizable systems. The stability analysis of time-varying systems in Lemma~\ref{lem:Stability} is a fundamental tool. To exploit it, we first provide a uniform bound on the Riccati iterations as an explicit function of $P_\star$. 
\begin{lem} \label{lem:Puppber}
Let $P \in \mathbb{S}^n_+$ satisfy $\|P-P_\star\|\leqslant \varepsilon_{\mathrm{p}}$ and $P_\star^{(k)}$ denote $\mathcal{R}^{(k)}_{A_\star,B_\star}(P)$. If Assumptions~\ref{ass:LQR}, \ref{ass:LQRgeneral2} hold, then 
\begin{equation} \label{eq:BoundRiccUniformSta}    
P^{(k)}_\star \preceq	   \big[ \|P_\star\|+ \varepsilon_{\mathrm{p}}\beta_\star^{-1} \big] I, \text{ }\forall \text{ }k \in \mathbb{Z}^+.
\end{equation}
\end{lem}

 \begin{proof}
 Given $P_\star^{(k)}$, recall $\bar{x}^\top P_\star^{(k+1)} \bar{x} = \min_{\bar{u}} l(\bar{x},\bar{u}) + (\bar{x}^+)^\top P_\star^{(k)} \bar{x}^+$ for a nominal system $\bar{x}^+ = A_\star \bar{x} + B_\star \bar{u}$ and any $\bar{x} \in \mathbb{R}^n$. This implies $P_\star^{(k+1)} \preceq Q+ K_\star^\top R K_\star+ L_\star^\top P_\star^{(k)}L_\star$, i.e., the infinite-horizon LQR may not be finite-horizon optimal. Then using Lemma~\ref{lem:RicaDiffEqu}\ref{eq:IdentiRicatiLyapu} leads to
\begin{align}
&P_\star^{(k+1)}  \preceq	Q+ K_\star^\top R K_\star+ L_\star^\top P_\star^{(k)}L_\star - P_\star + P_\star    \nonumber  \\
&\implies P_\star^{(k+1)}  \preceq	L_\star^\top (P_\star^{(k)} - P_\star)L_\star  + P_\star    \nonumber  \\
&\implies P_\star^{(k+1)} \preceq    (L_\star^{k+1})^\top  (P-P_\star) L_\star^{k+1}+P_\star   \nonumber \\
& \implies \|P_\star^{(k+1)}\| \leqslant  \|P_\star\|+ \varepsilon_{\mathrm{p}}(\beta_\star^{-1}-1),  \nonumber
\end{align}
where we have used  \eqref{eq:RateLstar} in the last step, and we further simplify the bound's expression in \eqref{eq:BoundRiccUniformSta}.
 \end{proof}

\begin{lem} \label{lem:stabiTRANSI}
 For any $i,j \in \mathbb{Z}^+$ with $i\geqslant j$, if Assumptions~\ref{ass:LQR} and \ref{ass:LQRgeneral2} hold, then $\| \Phi^{(j:i)}_{A_\star,B_\star}(P) \| \leqslant 
 \Upsilon_\star^3 (1+ \overline{\sigma}(P_\star)+ \varepsilon_{\mathrm{p}}) $    
 \begin{align} 
 \times \sqrt{ \Big(\overline{\beta}(\varepsilon_{\mathrm{p}}) \Big ) ^{-1} \Big(1- \overline{\beta}(\varepsilon_{\mathrm{p}}) \Big)^{-1} }  \Big( \sqrt{1 -  \overline{\beta} (\varepsilon_{\mathrm{p}}) } \Big)^{{i-j}}.\label{eq:TransitionBoundSta}
\end{align} 
If Assumption~\ref{ass:TechContPertur}\ref{ass4:i} holds additionally, then 
\begin{align} \label{eq:PerturbationStabilizable}
\| \Phi^{(j:i)}_{\hat{A},\hat{B}}(P) \| \leqslant 
\overline{\zeta}  \Big( \sqrt{ 1 -  \alpha_{\varepsilon_\mathrm{m}}^{-1} \overline{\beta} (\varepsilon_{\mathrm{p}}) } \Big)^{{i-j}},
\end{align}
with $\overline{\zeta}$ defined in \eqref{eq:PhiConstant2}.
\end{lem} 
\begin{proof}
    As $\{P_\star^{(k)} \}_{k=j}^{i-1}$ is a positive-definite matrix sequence when $j \geqslant 1$ and satisfies $P_\star^{(k)} \succeq Q \succeq I$, we first focus on this case where $j\geqslant 1$. Then for $ k\in [j,i]$ and from Lemma~\ref{lem:RicaDiffEqu}\ref{eq:IdentiRicatiLyapu}, $V(k, \bar{x}) = \bar{x}^\top P_\star^{(i+j-k)} \bar{x}$ is a time-varying Lyapunov function, satisfying $V(k+1, \bar{x}_{k+1}) - V(k, \bar{x}_{k}) \leqslant -\bar{x}_{k}^\top Q  \bar{x}_{k}  $, for the time-varying nominal system $\bar{x}_{k+1} = \mathcal{L}_{A_\star,B_\star}\big( P_\star^{(i+j-1-k)}\big)  \bar{x}_{k}$. Note that this nominal system admits the state-transition matrix of interest, i.e., $\bar{x}_i = \Phi^{(j:i)}_{A_\star,B_\star}(P) \bar{x}_j$. Then when $j \geqslant 1$, we have $\underline{\sigma}(Q)\|x\|^2 \leqslant V(k, x) \leqslant   \big[ \|P_\star\|+ \varepsilon_{\mathrm{p}}\beta_\star^{-1} \big] \|x\|^2$ based on \eqref{eq:BoundRiccUniformSta}. Applying Lemma~\ref{lem:Stability} shows
\begin{align}   \label{eq:PhiStab}
    &\| \Phi^{(j:i)}_{A_\star,B_\star}(P) \| \leqslant  \sqrt{ \overline{\beta}^{-1} } \Big( \sqrt{1 -  \overline{\beta}} \Big)^{{i-j}}, \text{ for } j \geqslant 1,  
\end{align}
where we recall $\overline{\beta}(\varepsilon_{\mathrm{p}}) = \underline{\sigma}(Q)/ \big[ \|P_\star\|+ \varepsilon_{\mathrm{p}}\beta_\star^{-1} \big]$. Eq.~\eqref{eq:PhiStab} proves \eqref{eq:TransitionBoundSta} for $j \geqslant 1$ due to $\Upsilon_\star^3 (1+ \overline{\sigma}(P_\star)+ \varepsilon_{\mathrm{p}}) \sqrt{ \Big(1- \overline{\beta}(\varepsilon_{\mathrm{p}}) \Big)^{-1} } \geqslant 1 $.

However, when $j=0$, recall that $P = P_\star^{(0)}$ is positive semi-definite by Assumption~\ref{ass:LQR}, and thus it cannot be lower bounded by a positive definite matrix. Then we exploit 
\begin{equation} \label{eq:proofNo}
  \| \Phi^{(0:i)}_{A_\star,B_\star}(P) \| \leqslant \|\mathcal{L}_{A_\star,B_\star} (P) \| \| \Phi^{(1:i)}_{A_\star,B_\star}(P)\|,  
\end{equation}
for $i \geqslant 1$ and
\begin{align}
&\| \mathcal{L}_{A_\star,B_\star} (P) \|  \leqslant \|A_\star \| + \|B_\star\| \|\mathcal{K}_{A_\star,B_\star}(P_\star^{(0 )}  
 )\| \nonumber \\
 & \leqslant  \|A_\star \| + \|B_\star\|^2 \|A_\star\| \|P\| \leqslant \Upsilon_\star^3 (1+ \|P\|). \label{eq:BoundCloseLoop}
\end{align}
Combining the above bound, \eqref{eq:PhiStab}, \eqref{eq:proofNo} and $\|P\| = \|P-P_\star+P_\star\| \leqslant  \varepsilon_{\mathrm{p}}+ \|P_\star\|$ proves \eqref{eq:TransitionBoundSta} for $j=0$. The case $i=0$ holds trivially.

Then if $(\hat{A}, \hat{B})$ is considered instead, $ \| \Phi^{(j:i)}_{\hat{A}, \hat{B} }(P) \| $ can be upper bounded as in \eqref{eq:TransitionBoundSta}, with $P_\star$ replaced by $\hat{P}$. Then combining this bound with $\| \widehat{P} \| \leqslant \alpha_{\varepsilon_\mathrm{m}} \| P_\star\|$ from Lemma~\ref{lem:PerturbationSimchow} and $ \alpha_{\varepsilon_\mathrm{m}}  \in [1,\sqrt{2}]$ from \eqref{eq:BoundALpha} implies \eqref{eq:PerturbationStabilizable} with
\begin{IEEEeqnarray}{rCl}
 \overline{\zeta} &\triangleq &    (2 \Upsilon_\star )^3 (1+ \sqrt{2} \|P_\star\|+ \Upsilon_\star  ) \nonumber \\
&&\negmedspace{} \times \sqrt{  \sqrt{2} (\|P_\star\|+ \Upsilon_\star \beta_\star^{-1} ) (1 -   \beta_\star )^{-1}    },   \label{eq:PhiConstant2}
\end{IEEEeqnarray}
where we have used \eqref{eq:BoundALpha}, the fact $\Big( \sqrt{1 -  \overline{\beta} } \Big)^{{i-j-1}} \leqslant \Big( \sqrt{1 -  \alpha_{\varepsilon_\mathrm{m}}^{-1} \overline{\beta}} \Big)^{{i-j-1}} $ and $\Upsilon_\star \geqslant \max\{\varepsilon_{\mathrm{m}},\varepsilon_{\mathrm{p}} \}$.
 \end{proof}
 \begin{rem}
The stability of $\Phi^{(0:i)}_{A_\star,B_\star}(P)$ is a classical result, e.g., see \cite{bitmead1991riccati} and \cite[Lem. 2.9]{anderson1986stability}. The main challenge in Lemma~\ref{lem:stabiTRANSI} is to derive the decay rate explicitly as a function of $P_\star$, such that the perturbation result in \eqref{eq:PerturbationStabilizable} is obtained.
\end{rem}

Finally, Theorem~\ref{thm:stabi} is proved by combining Lemmas~\ref{lem:Pdiff}, \ref{lem:PhiBar}, and \ref{lem:stabiTRANSI}.

\bibliographystyle{IEEEtran}
\bibliography{LearningControl}  

\begin{thebibliography}{10}
\providecommand{\url}[1]{#1}
\csname url@samestyle\endcsname
\providecommand{\newblock}{\relax}
\providecommand{\bibinfo}[2]{#2}
\providecommand{\BIBentrySTDinterwordspacing}{\spaceskip=0pt\relax}
\providecommand{\BIBentryALTinterwordstretchfactor}{4}
\providecommand{\BIBentryALTinterwordspacing}{\spaceskip=\fontdimen2\font plus
\BIBentryALTinterwordstretchfactor\fontdimen3\font minus
  \fontdimen4\font\relax}
\providecommand{\BIBforeignlanguage}[2]{{%
\expandafter\ifx\csname l@#1\endcsname\relax
\typeout{** WARNING: IEEEtran.bst: No hyphenation pattern has been}%
\typeout{** loaded for the language `#1'. Using the pattern for}%
\typeout{** the default language instead.}%
\else
\language=\csname l@#1\endcsname
\fi
#2}}
\providecommand{\BIBdecl}{\relax}
\BIBdecl

\bibitem{schwenzer2021review}
M.~Schwenzer, M.~Ay, T.~Bergs, and D.~Abel, ``Review on model predictive
  control: An engineering perspective,'' \emph{Int. J. Adv. Manuf. Technol.},
  vol. 117, no. 5-6, pp. 1327--1349, 2021.

\bibitem{katayama2023model}
S.~Katayama, M.~Murooka, and Y.~Tazaki, ``Model predictive control of legged
  and humanoid robots: models and algorithms,'' \emph{Adv. Robot.}, vol.~37,
  no.~5, pp. 298--315, 2023.

\bibitem{mayne2014model}
D.~Q. Mayne, ``Model predictive control: Recent developments and future
  promise,'' \emph{Automatica}, vol.~50, no.~12, pp. 2967--2986, 2014.

\bibitem{grune2011Book}
L.~Gr{\"u}ne and J.~Pannek, \emph{Nonlinear {M}odel {P}redictive {C}ontrol:
  {T}heory and {A}lgorithms}.\hskip 1em plus 0.5em minus 0.4em\relax Springer,
  2011.

\bibitem{grune2020economic}
L.~Gr{\"u}ne and S.~Pirkelmann, ``Economic model predictive control for
  time-varying system: Performance and stability results,'' \emph{Optim.
  Control Appl. Methods}, vol.~41, no.~1, pp. 42--64, 2020.

\bibitem{kohler2021stability}
J.~K{\"o}hler, M.~Zeilinger, and L.~Gr{\"u}ne, ``Stability and performance
  analysis of {NMPC}: Detectable stage costs and general terminal costs,''
  \emph{IEEE Trans. Autom. Control}, vol.~68, no.~10, pp. 6114--6129, 2023.

\bibitem{bertsekas2022newton}
D.~Bertsekas, ``Newton’s method for reinforcement learning and model
  predictive control,'' \emph{Results Control Optim.}, vol.~7, p. 100121, 2022.

\bibitem{moreno2023predictive}
F.~Moreno-Mora, L.~Beckenbach, and S.~Streif, ``Predictive control with
  learning-based terminal costs using approximate value iteration,''
  \emph{IFAC-PapersOnLine}, vol.~56, no.~2, pp. 3874--3879, 2023.

\bibitem{bertsekas2019reinforcement}
D.~Bertsekas, \emph{Reinforcement {L}earning and {O}ptimal {C}ontrol}.\hskip
  1em plus 0.5em minus 0.4em\relax Athena Scientific, 2019.

\bibitem{bitmead1991riccati}
R.~R. Bitmead and M.~Gevers, ``Riccati difference and differential equations:
  Convergence, monotonicity and stability,'' in \emph{The Riccati
  Equation}.\hskip 1em plus 0.5em minus 0.4em\relax Springer, 1991, pp.
  263--291.

\bibitem{moreno2022performance}
F.~Moreno-Mora, L.~Beckenbach, and S.~Streif, ``Performance bounds of adaptive
  {MPC} with bounded parameter uncertainties,'' \emph{Eur. J. Control},
  vol.~68, p. 100688, 2022.

\bibitem{muthirayan2021online}
D.~Muthirayan, J.~Yuan, D.~Kalathil, and P.~P. Khargonekar, ``Online learning
  for predictive control with provable regret guarantees,'' in \emph{Proc. 61st
  IEEE Conf. Decis. Control}, 2022, pp. 6666--6671.

\bibitem{lale2021model}
S.~Lale, K.~Azizzadenesheli, B.~Hassibi, and A.~Anandkumar, ``Model learning
  predictive control in nonlinear dynamical systems,'' in \emph{Proc. 60th IEEE
  Conf. Decis. Control}, 2021, pp. 757--762.

\bibitem{dogan2021regret}
I.~Dogan, Z.-J.~M. Shen, and A.~Aswani, ``Regret analysis of learning-based
  {MPC} with partially-unknown cost function,'' \emph{IEEE Trans. Autom.
  Control}, vol.~69, no.~5, pp. 3246--3253, 2024.

\bibitem{recht2019tour}
B.~Recht, ``A tour of reinforcement learning: {T}he view from continuous
  control,'' \emph{Annu. Rev. Control Robot. Auton. Syst.}, vol.~2, pp.
  253--279, 2019.

\bibitem{matni2019self}
N.~Matni, A.~Proutiere, A.~Rantzer, and S.~Tu, ``From self-tuning regulators to
  reinforcement learning and back again,'' in \emph{Proc. 58th IEEE Conf.
  Decis. Control}, 2019, pp. 3724--3740.

\bibitem{tsiamis2023statistical}
A.~Tsiamis, I.~Ziemann, N.~Matni, and G.~J. Pappas, ``Statistical learning
  theory for control: {A} finite-sample perspective,'' \emph{IEEE Control
  Syst}, vol.~43, no.~6, pp. 67--97, 2023.

\bibitem{dean2020sample}
S.~Dean, H.~Mania, N.~Matni, B.~Recht, and S.~Tu, ``On the sample complexity of
  the linear quadratic regulator,'' \emph{Found. Comput. Math.}, vol.~20,
  no.~4, pp. 633--679, 2020.

\bibitem{mania2019certainty}
H.~Mania, S.~Tu, and B.~Recht, ``Certainty equivalence is efficient for linear
  quadratic control,'' in \emph{Proc. 33rd Adv. Neural Inf. Process. Syst.},
  2019, pp. 10\,154--10\,164.

\bibitem{simchowitz2020naive}
M.~Simchowitz and D.~Foster, ``Naive exploration is optimal for online {LQR},''
  in \emph{Proc. 37th Int. Conf. Mach. Learn.}, 2020, pp. 8937--8948.

\bibitem{del2021note}
P.~Del~Moral and E.~Horton, ``A note on {R}iccati matrix difference
  equations,'' \emph{SIAM J. Control Optim.}, vol.~60, no.~3, pp. 1393--1409,
  2022.

\bibitem{cormen2022introduction}
T.~H. Cormen, C.~E. Leiserson, R.~L. Rivest, and C.~Stein, \emph{Introduction
  to {A}lgorithms}.\hskip 1em plus 0.5em minus 0.4em\relax MIT press, 2022.

\bibitem{simchowitz2018learning}
M.~Simchowitz, H.~Mania, S.~Tu, M.~I. Jordan, and B.~Recht, ``Learning without
  mixing: {T}owards a sharp analysis of linear system identification,'' in
  \emph{Prof. 31st Conference on Learning Theory}, 2018, pp. 439--473.

\bibitem{sarkar2019near}
T.~Sarkar and A.~Rakhlin, ``Near optimal finite time identification of
  arbitrary linear dynamical systems,'' in \emph{Proc. 36th Int. Conf. Mach.
  Learn.}, 2019, pp. 5610--5618.

\bibitem{bertsekas2012dynamic}
D.~Bertsekas, \emph{Dynamic {P}rogramming and {O}ptimal {C}ontrol: Volume
  I}.\hskip 1em plus 0.5em minus 0.4em\relax Athena scientific, 2012.

\bibitem{anderson1979optimal}
B.~Anderson and J.~B. Moore, \emph{Optimal {F}iltering}.\hskip 1em plus 0.5em
  minus 0.4em\relax Prentice-Hall, 1979.

\bibitem{lawson2007birkhoff}
J.~Lawson and Y.~Lim, ``A {B}irkhoff contraction formula with applications to
  {R}iccati equations,'' \emph{SIAM J. Control Optim.}, vol.~46, no.~3, pp.
  930--951, 2007.

\bibitem{hassibi1999indefinite}
B.~Hassibi, A.~H. Sayed, and T.~Kailath, \emph{Indefinite-Quadratic
  {E}stimation and {C}ontrol: {A} {U}nified {A}pproach to $H_2$ and $H_\infty$
  {T}heories}.\hskip 1em plus 0.5em minus 0.4em\relax SIAM, 1999.

\bibitem{rugh1996linear}
W.~J. Rugh, \emph{Linear System Theory}.\hskip 1em plus 0.5em minus 0.4em\relax
  Prentice-Hall, 1996.

\bibitem{zhang2021regret}
R.~Zhang, Y.~Li, and N.~Li, ``On the regret analysis of online {LQR} control
  with predictions,'' in \emph{Proc. 2021 Am. Control Conf.}, 2021, pp.
  697--703.

\bibitem{konstantinov1995conditioning}
M.~Konstantinov, I.~Poptech, and V.~Angelova, ``Conditioning and sensitivity of
  the difference matrix {R}iccati equation,'' in \emph{Proc. 1995 Am. Control
  Conf.}, 1995, pp. 466--466.

\bibitem{Schutter2001model}
B.~De~Schutter and T.~Van Den~Boom, ``Model predictive control for
  max-plus-linear discrete event systems,'' \emph{Automatica}, vol.~37, no.~7,
  pp. 1049--1056, 2001.

\bibitem{fruchard2012choice}
M.~Fruchard, G.~Allibert, and E.~Courtial, ``Choice of the control horizon in
  an {NMPC} strategy for the full-state control of nonholonomic systems,'' in
  \emph{Proc. 2012 Am. Control Conf.}\hskip 1em plus 0.5em minus 0.4em\relax
  IEEE, 2012, pp. 4149--4154.

\bibitem{lale2022reinforcement}
S.~Lale, K.~Azizzadenesheli, B.~Hassibi, and A.~Anandkumar, ``Reinforcement
  learning with fast stabilization in linear dynamical systems,'' in
  \emph{Proc. 25th Int. Conf. Artif. Intell. Stat.}, 2022, pp. 5354--5390.

\bibitem{yu2020power}
C.~Yu, G.~Shi, S.~J. Chung, Y.~Yue, and A.~Wierman, ``The power of predictions
  in online control,'' in \emph{Proc. 34th Adv. Neural Inf. Process. Syst.},
  2020, pp. 1994--2004.

\bibitem{athrey2024regret}
A.~Athrey, O.~Mazhar, M.~Guo, B.~De~Schutter, and S.~Shi, ``Regret analysis of
  learning-based linear quadratic {G}aussian control with additive
  exploration,'' in \emph{Proc. 22nd European Control Conference}, 2024, pp.
  1795--1801.

\bibitem{krener2018adaptive}
A.~J. Krener, ``Adaptive horizon model predictive control,''
  \emph{IFAC-PapersOnLine}, vol.~51, no.~13, pp. 31--36, 2018.

\bibitem{bohn2023optimization}
E.~B{\o}hn, S.~Gros, S.~Moe, and T.~A. Johansen, ``Optimization of the model
  predictive control meta-parameters through reinforcement learning,''
  \emph{Eng. Appl. Artif. Intell.}, vol. 123, p. 106211, 2023.

\bibitem{chen2022real}
S.~Chen, K.~Werling, A.~Wu, and C.~K. Liu, ``Real-time model predictive control
  and system identification using differentiable simulation,'' \emph{IEEE
  Robot. Autom. Lett.}, vol.~8, no.~1, pp. 312--319, 2022.

\bibitem{liu2024certainty}
C.~Liu, S.~Shi, and B.~De~Schutter, ``Certainty-equivalence model predictive
  control: Stability, performance, and beyond,'' \emph{arXiv preprint
  arXiv:2412.10625}, 2024.

\bibitem{anderson1986stability}
B.~D.~O. Anderson \emph{et~al.}, \emph{Stability of {A}daptive {S}ystems:
  {P}assivity and {A}veraging {A}nalysis}.\hskip 1em plus 0.5em minus
  0.4em\relax MIT Press, 1986.

\end{thebibliography}
 
\end{document}